\documentclass[10pt, latterpaper, oneside, onecolumn ]{article}   	%
\usepackage[utf8]{inputenc}
\usepackage{geometry}                	
\usepackage{tikz}
\usetikzlibrary{positioning}
\usetikzlibrary{arrows}
\tikzstyle{block}=[draw opacity=0.7,line width=1.4cm]
\usetikzlibrary{positioning,shapes,shadows,arrows,scopes,decorations}
\usetikzlibrary{matrix,chains,shapes.geometric,shapes.arrows,automata,chains,er,fit}

\tikzstyle{comment}=[rectangle, draw=black, fill=red!50!black, rounded corners, drop shadow,
anchor=west, text=white]
\usetikzlibrary{shapes.callouts}
\usetikzlibrary{mindmap}
\usetikzlibrary{calc,through,backgrounds}

\usepackage{amsmath,amssymb}
\usepackage{cite}
\usepackage{graphicx}
\usepackage{dsfont}			
\usepackage{bbm}
\usepackage{calc,pgf,xcolor}

\usepackage{enumerate}
\usepackage{amsthm}
\usepackage{float}

\newtheorem{theorem}{Theorem}
\newtheorem{proposition}{Proposition}
\newtheorem{conjecture}{Conjecture}

\newtheorem{lemma}{Lemma}

\newtheorem{definition}{Definition}

\theoremstyle{remark}

\newtheorem{remark}{Remark}

\def\dfn{\stackrel{\textrm{def}}{=}}

\def\bern{\mathsf{Bern}}
\def\cgk{C_{\mathsf{GK}}}
\def\cw{C_{\mathsf{W}}}
\def\tn{{\otimes n}}
\def\dtv{d_{\mathsf{TV}}}
\def\dsbs{\mathsf{DSBS}}

\DeclareMathOperator{\E}{\mathbb{E}}
\DeclareMathOperator{\ind}{\mathds{1}}

\DeclareMathOperator{\Var}{\mathsf{Var}}
\DeclareMathOperator{\Cov}{\mathsf{Cov}}
\DeclareMathOperator{\rank}{rank}

\title{Distributed Source Simulation With No Communication}
\author{Tomer~Berg, Ofer~Shayevitz, Young-Han~Kim and Lele~Wang  \thanks{This work has been supported by an ERC grant no. 639573 and an ISF grant no. 1367/14. \newline T. Berg and O. Shayevitz  are with the Department of Electrical Engineering - Systems, Tel Aviv University, Tel Aviv, Israel (email: tomerberg@mail.tau.ac.il, ofersha@eng.tau.ac.il). \newline Young-Han Kim is with the Department of Electrical and Computer Engineering, University of California, San Diego, La Jolla, CA 92093 USA (email:yhk@ucsd.edu). \newline Lele Wang is with the Department of Electrical and Computer Engineering, the University of British Columbia, 2332 Main Mail, Vancouver, BC, V6T 1Z4, Canada (email: lelewang@ece.ubc.ca). }}
\date{}

\begin{document}

\allowdisplaybreaks 	
	\maketitle
	
	\begin{abstract}
	We consider the  problem of distributed source simulation with no communication, in which Alice and Bob observe sequences $U^n$ and $V^n$ respectively, drawn from a joint distribution $p_{UV}^\tn$, and wish to locally generate sequences $X^n$ and $Y^n$ respectively with a joint distribution that is close (in KL divergence) to $p_{XY}^\tn$. We provide a single-letter condition under which such a simulation is asymptotically possible with a vanishing KL divergence. Our condition is nontrivial only in the case where the G{\`a}cs-K{\"o}rner (GK) common information between $U$ and $V$ is nonzero, and we conjecture that only scalar Markov chains $X-U-V-Y$ can be simulated otherwise. Motivated by this conjecture, we further examine the case where both $p_{UV}$ and $p_{XY}$ are doubly symmetric binary sources with parameters $p,q\leq 1/2$ respectively. While it is trivial that in this case $p\leq q$ is both necessary and sufficient, we show that when $p$ is close to $q$ then any successful simulation is close to being scalar in the total variation sense. 
	\end{abstract}
	
	\section{Introduction and Main Results}
	Let us consider the following distributed simulation problem. Assume that $(U^n,V^n)$ are drawn by nature according to some i.i.d. distribution $p_{UV}$. Alice has access to $U^n$ and she outputs some sequence $X^n$, while Bob has access to $V^n$ and he outputs some sequence $Y^n$, such that $(X^n,Y^n)$ are approximately distributed according to some i.i.d. distribution $p_{XY}$. There is no communication between the parties nor do they share any common randomness (this setup is depicted in Fig.~\ref{fig:DSS_setup}). Our goal is to characterize the set of distributions $p_{XY}$ that can be reliably simulated using this scheme. 
\begin{center}
	\begin{figure}[H]
		\centering
		\resizebox*{!}{0.3\textwidth}{
			\begin{tikzpicture}[node distance=1cm, auto]  
			\tikzset{
			mynode/.style={rectangle,draw=black, top color=white, bottom color=white, thick, minimum width=1cm, minimum height = 1cm, text centered},
			myarrow1/.style={<-, >=latex', shorten >=1pt, line width=0.4mm},
			myarrow2/.style={->, >=latex', shorten >=1pt, line 	width=0.4mm},	
			mylabel/.style={text width=5em, text centered}  
			}   
			\draw (2,-1.6) node (XY_Dist) {$(X^n,Y^n)\underset{\tiny{approx}}{\sim} p_{XY}^\tn$};
			\draw (2,1.8) node (UV_Dist) {$(U^n,V^n)\sim p_{UV}^\tn$ };
			\draw (-0.8,1) node (U) {$U^n$};
			\draw (4.8,1) node (V) {$V^n$};
			\draw (-2.6,0) node (X) {$X^n$};
			\draw (6.5,0) node (Y) {$Y^n$};
			\draw (-0.5,0) node [mynode,text width=4em] (encoder_A) {$p_{X^n|U^n}$};
			\draw (4.5,0) node [mynode,text width=4em] (encoder_B) {$p_{Y^n|V^n}$};
			\draw (-0.5,-0.8) node (Alice) {Alice};
			\draw (4.5,-0.8) node (Bob) {Bob};
			
			\draw[myarrow1] (encoder_A.north) |-  (UV_Dist.west);
			\draw[myarrow1](encoder_B.north) |-  (UV_Dist.east);
			\draw[myarrow2] (encoder_A.west) |- (X.east);
			\draw[myarrow2] (encoder_B.east) |- (Y.west);			
			\end{tikzpicture} 
		}
		\caption{Distributed source simulation}
		\label{fig:DSS_setup}
	\end{figure}
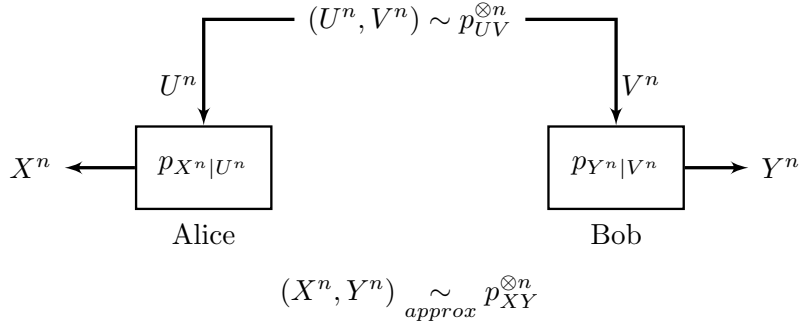
\end{center}
	To make this more formal, Let $p_{UV}$ be some joint discrete distribution, and let $(U^n,V^n) \sim p_{UV}^\tn$. We say that a joint distribution {\em $p_{XY}$ is $(n,\epsilon)$-simulable from $p_{UV}$}, if there exist conditional probability distributions $p_{X^n|U^n}$ and $p_{Y^n|V^n}$ such that the distribution 
	\begin{align*}
	&p_{X^nY^n}(x^n,y^n) \\&= \sum_{u^n,v^n} p_{UV}^\tn(u^n,v^n)p_{X^n|U^n}(x^n|  u^n) p_{Y^n|V^n}(y^n|  v^n)
	\end{align*}
	is $\epsilon$-close in relative entropy to $p_{XY}^\tn$, i.e., 
	\begin{align}
	D\left(p_{X^nY^n} \| p_{XY}^\tn\right) \leq \epsilon.
	\end{align} 
	We say that {\em $p_{XY}$ is simulable from $p_{UV}$} if it is $(n,\epsilon)$-simulable from $p_{UV}$ for every $\epsilon>0$ and $n$ sufficiently large.  
	\begin{remark}
	In our setup, we require Alice and Bob to generate one sample from $p_{XY}$ per each sample of $p_{UV}$. One can also consider a general conversion rate of $\alpha$, where Alice and Bob remotely use $U^n,V^n$ in order to generate a distribution that is approximately $p_{XY}^{\otimes \lceil \alpha n\rceil}$. However, this case is essentially equivalent to our setup. If $\alpha = k/m$ for some integers $k,m$, then one can look at the rate one problem with source distribution $P_{UV}^{\otimes m}$ and target distribution $P_{XY}^{\otimes k}$. 
	\end{remark}  
	For $U=V$, our question was already answered by Wyner:
	\begin{theorem}[ \cite{wyner1975common}]
    If $H(U) \geq \cw(X;Y) $, where
	\begin{align*}
	\cw(X;Y) \dfn \min_{W:X-W-Y} I(X,Y;W)
	\end{align*}
    then $p_{XY}$ is simulable from $p_U$.
    \end{theorem} 
    Wyner originally considered the case where $U\sim \bern(1/2)$, but the extension above is pretty obvious. $\cw(X;Y)$ is the so-called {\em Wyner common information}, defined as the minimum number of common i.i.d. random bits that must be supplied to Alice and Bob in order for them to be able to locally create sequences $X^n$ and $Y^n$ respectively, where $p_{X^n,Y^n}$ is arbitrarily close (in total variation, or in KL divergence) to being i.i.d. $p_{XY}^\tn$, in the limit of large $n$. Note that this solution is "digital", in the sense that it uses codebooks. One naive approach that comes to mind is a reduction to Wyner's setup, by generating a ``common part'' $f(U) = g(V)$ from $U\neq V$. This corresponds to using the so-called {\em G{\`a}cs-K{\"o}rner (GK) common information}~\cite{gacs1973common}, which is defined as 
	\begin{align}
	\cgk(U;V) \dfn \max_{f,g:\,\Pr(f(U)=g(V)) = 1}H(f(U)).
	\end{align}
	$\cgk(U;V)$ is the maximum amount of randomness that can be agreed upon by two separate agents, Alice and Bob, observing $U$ or $V$ respectively. The (unique) random variable $K=f(U)=g(V)$ that attains the maximum above is called the {\em GK common part} of $(U,V)$. It is well known that the GK common information tensorizes, in the sense that $\cgk(X^n; Y^n) = n\cgk(X;Y)$ where $(X^n,Y^n)\sim p_{XY}^\tn$. In other words, the GK common part of $(X^n,Y^n)$ is simply the vector of scalar common parts pertaining to each $(X_i,Y_i)$. Moreover, this tensorization is stable in the sense that it remains asymptotically valid even if a vanishing error is allowed~\cite{gacs1973common,witsenhausen1975sequences}.
	
	Combining the two results, Alice and Bob can both extract the GK common part $K^n$ from $U^n$ and $V^n$ respectively, and use Wyner coding, which leads to the following simple solution:
    \begin{proposition}[digital solution]\label{prop:digital}
    If $H(K)=\cgk(U;V) \geq \cw(X;Y)$, then $p_{XY}$ is simulable from $p_{UV}$.
    \end{proposition}
    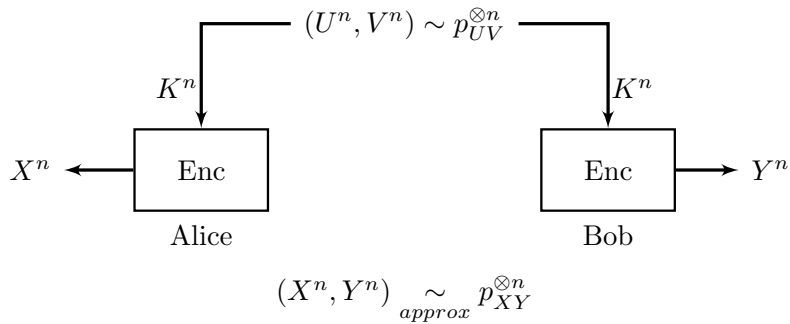
\begin{figure}[H]
	\centering
	\resizebox*{!}{0.296\textwidth}{
		\begin{tikzpicture}[node distance=1cm, auto]  
		\tikzset{
			mynode/.style={rectangle,draw=black, top color=white, bottom color=white, thick, minimum width=1cm, minimum height = 1cm, text centered},
			myarrow1/.style={<-, >=latex', shorten >=1pt, line width=0.4mm},
			myarrow2/.style={->, >=latex', shorten >=1pt, line width=0.4mm},	
			mylabel/.style={text width=5em, text centered}  
		}  
		\draw (2,-1.6) node (XY_Dist) {$(X^n,Y^n)\underset{\tiny{approx}}{\sim} p_{XY}^\tn$};
		\draw (2,1.8) node (UV_Dist) {$(U^n,V^n)\sim p_{UV}^\tn$ };
		\draw (-0.8,1) node (K1) {$K^n$};
		\draw (4.8,1) node (K2) {$K^n$};
		\draw (-2.6,0) node (X) {$X^n$};
		\draw (6.5,0) node (Y) {$Y^n$};
		\draw (-0.5,0) node [mynode,text width=4em] (encoder_A) {Enc};
		\draw (4.5,0) node [mynode,text width=4em] (encoder_B) {Enc};
		\draw (-0.5,-0.8) node (Alice) {Alice};
		\draw (4.5,-0.8) node (Bob) {Bob};
		
		\draw[myarrow1] (encoder_A.north) |-  (UV_Dist.west);
		\draw[myarrow1](encoder_B.north) |-  (UV_Dist.east);
		\draw[myarrow2] (encoder_A.west) |- (X.east);
		\draw[myarrow2] (encoder_B.east) |- (Y.west);
		
		\end{tikzpicture} 
	}
	\caption{Digital solution}
	\label{fig:Digital}
\end{figure}
    This digital approach is viable only when $\cgk(U;V)>0$. There is an even simpler analog approach that does not use common information -- Alice and Bob pass their corresponding sequences through memoryless channels $p_{X^n|U^n}=p_{X|U}^\tn$ and $p_{Y^n|V^n}=p_{Y|V}^\tn$, respectively, symbol-by-symbol.
    \begin{proposition}[analog solution]\label{prop:analog}
    If $X-U-V-Y$ form a Markov chain, then $p_{XY}$ is simulable from $p_{UV}$.
    \end{proposition}
	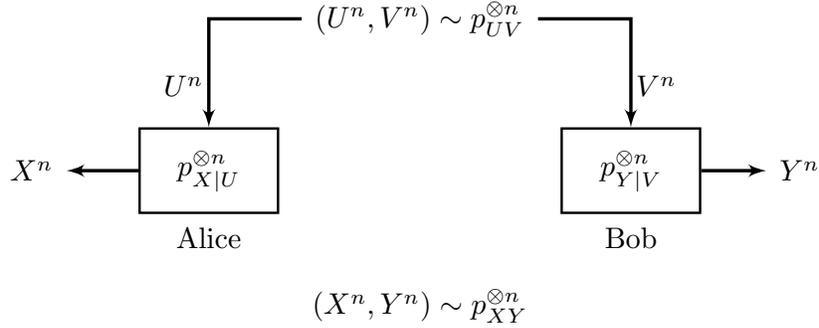
\begin{figure}[H]
	\centering
	\resizebox*{!}{0.3\textwidth}{
		\begin{tikzpicture}[node distance=1cm, auto]  
		\tikzset{
			mynode/.style={rectangle,draw=black, top color=white, bottom color=white, thick, minimum width=1cm, minimum height = 1cm, text centered},
			myarrow1/.style={<-, >=latex', shorten >=1pt, line width=0.4mm},
			myarrow2/.style={->, >=latex', shorten >=1pt, line width=0.4mm},	
			mylabel/.style={text width=5em, text centered}  
		}  
		\draw (2,-1.6) node (XY_Dist) {$(X^n,Y^n)\sim p_{XY}^\tn$};
		\draw (2,1.8) node (UV_Dist) {$(U^n,V^n)\sim p_{UV}^\tn$ };
		\draw (-0.8,1) node (U) {$U^n$};
		\draw (4.8,1) node (V) {$V^n$};
		\draw (-2.6,0) node (X) {$X^n$};
		\draw (6.5,0) node (Y) {$Y^n$};
		\draw (-0.5,0) node [mynode,text width=4em] (encoder_A) {$p_{X|U}^\tn$};
		\draw (4.5,0) node [mynode,text width=4em] (encoder_B) {$p_{Y|V}^\tn$};
		\draw (-0.5,-0.8) node (Alice) {Alice};
		\draw (4.5,-0.8) node (Bob) {Bob};
		
		\draw[myarrow1] (encoder_A.north) |-  (UV_Dist.west);
		\draw[myarrow1](encoder_B.north) |-  (UV_Dist.east);
		\draw[myarrow2] (encoder_A.west) |- (X.east);
		\draw[myarrow2] (encoder_B.east) |- (Y.west);
		
		\end{tikzpicture} 
	}
	\caption{Analog solution}
	\label{fig:Analog}
\end{figure}
	The first contribution of this work is the following characterization of a generally larger set of simulable distributions. 
	\begin{theorem}\label{thrm:achievable_region}
		Let $K$ be the GK common part of $(U,V)$. Suppose that 
		\begin{gather}
			\label{eq:markov1}W-K-(U,V)\\
			\label{eq:markov2}X-(U,W)-(V,W)-Y
		\end{gather} 
		are Markov chains, and 
		\begin{align}
			\cgk(U;V) \geq  I(X,Y;K,W).\label{eq:cgk_suff_cond}
		\end{align}
		Then $p_{XY}$ is simulable from $p_{UV}$.  	
	\end{theorem}
	Let $\mathcal{S}_{\mathrm{dig}}(p_{UV})$, $\mathcal{S}_{\mathrm{ana}}(p_{UV})$, and $\mathcal{S}(p_{UV})$ denote the collections all $p_{XY}$ simulable from $p_{UV}$ via a digital scheme (Proposition~\ref{prop:digital}), an analog scheme (Proposition~\ref{prop:analog}) and a hybrid scheme (Theorem~\ref{thrm:achievable_region}), respectively. The following proposition shows that the statement of Theorem~\ref{thrm:achievable_region} is not trivial. 
    \begin{proposition}\label{prop:better_than_cgk}
    $\mathcal{S}_{\mathrm{dig}}(p_{UV}) \cup \mathcal{S}_{\mathrm{ana}}(p_{UV}) \subseteq \mathcal{S}(p_{UV})$, and the inclusion is strict for some $p_{UV}$. Moreover, $\mathcal{S}(p_{UV})$ is strictly larger than $\mathcal{S}_{\mathrm{ana}}(p_{UV})$ if and only if $\cgk(U;V)>0$.
    \end{proposition}

In Theorem~\ref{thrm:achievable_region}, our agents' ability to cooperate stems from having some common information. No common part means no perfect cooperation, and this motivates us to conjecture that only analog simulation is possible in such a case.
	\begin{conjecture}
		If $\cgk(U;V)=0$ then $p_{XY}$ is simulable from $p_{UV}$ if and only if $X-U-V-Y$ forms a Markov chain. 
	\end{conjecture}

We are currently unable to prove or refute this conjecture. Note however that in some simple restricted cases the conjecture holds due to other impossibility results. For example, if $(U,V)$ is a $\dsbs(p)$ for some $p<1/2$ (hence $\cgk(U;V) = 0$) and we are only interested in simulating $\dsbs(q)$, then it is easy to see that $q\in[p,1-p]$ is both necessary and sufficient, and can be attained by a scalar Markov chain. Our next result shows that this is true in a stronger way; namely, when $q$ is close to $p$, then not only is the scalar Markov chain optimal, but it is essentially the only way to simulate a $\dsbs(q)$.

Let $\sigma$ be a permutation on $[n]$. With some abuse of notation, we refer to $\sigma$ as a {\em coordinate permutation} when applied to any $n$-vector, i.e., $\sigma(x^n) \dfn (x_{\sigma(1)}, x_{\sigma(2)},\ldots, x_{\sigma(n)})$. We write $\dtv(P,Q)$ to denote the total variation distance between the probability distributions $P$ and $Q$.
\begin{theorem}\label{thrm:DSBS}
	Let $(U,V)$ and $(X,Y)$ be $\dsbs(p)$ and $\dsbs(p+\delta)$ respectively, where $0\leq p \leq p+\delta \leq \frac{1}{2}$. Suppose that $p_{XY}$ is $(n,\epsilon)$-simulable from $p_{UV}$ via $p_{X^n|U^n}$ and $p_{Y^n|V^n}$. Then there exists a coordinate permutation $\sigma$ and scalar conditional distributions $q_{X_i|U_i}$ and $q_{Y_i|V_i}$ such that 
	\begin{align}	
	D_1 &\triangleq \dtv\left(p_{\sigma(X^n)|U^n}(\cdot\mid U^n)\,,\, \prod_{i=1}^nq_{X_i|U_i}(\cdot | U_i)\right)  \to  0,\label{eq:dtv_dsbs_u}\\
	D_2 &\triangleq \dtv\left(p_{\sigma(Y^n)|V^n}(\cdot\mid V^n)\,,\, \prod_{i=1}^nq_{Y_i|V_i}(\cdot | V_i)\right) \to  0, \label{eq:dtv_dsbs_v}
	\end{align}
	in probability, provided that $\epsilon,\delta=o(1/n)$. Conversely, if $\delta = \omega(1/\sqrt{n})$ or $\epsilon =\omega(1/\sqrt{n})$ then no such guarantee can be made, i.e., it is possible for $D_1,D_2$ to be bounded away from zero in probability as $n\to\infty$ for any scalar conditional distributions $q_{X_i|U_i}$ and $q_{Y_i|V_i}$.
\end{theorem}

Loosely speaking, the above result means that if $\delta$ and $\epsilon$ are small enough, then the actual mechanism under the hood of any successful simulation scheme is truly scalar, in the sense that no statistical test with access to the inputs of the mechanism can tell it apart from a scalar one. We note that Theorem~\ref{thrm:DSBS} is well known in combinatorics for the case where $\epsilon=\delta=0$ (see e.g.~\cite{frucht1949groups}), hence our result can be interpreted as a stable version of the aforementioned one. Furthermore, our result is close to being tight; when $\epsilon = \omega(1/\sqrt{n})$ or $\delta = \omega(1/\sqrt{n})$, successful simulation is possible using vector operations, for example by using other coordinates as noise.

\subsection*{Related work}
	In their classical paper on common randomness generation~\cite{ahlswede1993common}, Ahlswede and K{\"o}rner considered a setup in which Alice and Bob observe correlated i.i.d. r.v. pairs, and a noiseless channel with capacity $R$ from Alice to Bob is given. They defined the so-called {\em CR capacity} as the maximum entropy rate that Alice and Bob can agree upon with probability approaching one. The case of $R=0$ is related to our problem, but their setup is in some sense weaker since they only care about generating randomness, and not about simulating specific distributions. Cuff {\em et al}~\cite{cuff2010coordination} studied the joint distributions that can be generated by nodes in a network under communication constraints in which some of the nodes actions are randomly selected by nature. The predominant measure of successful simulation is the {\em empirical coordination}, which is defined as the total variation between the joint type of the actions and some prescribed distribution. Cover {\em et al}~\cite{cover2007capacity} characterized the empirical coordination needed for some 3-node problems. Abroshan {\em et al}~\cite{abroshan2015zero} considered an exact, zero error coordination instead of an asymptotically vanishing error and employed the notion of set coordination, which bears similarities with the empirical notion of coordination. The problem of channel simulation became a subject of interest in recent years. Soljanin~\cite{soljanin2002compressing} studied this in the context of quantum compression with unlimited common randomness. Bennett {\em et al} introduced a "reverse Shannon theorem"~\cite{bennett2002entanglement} (see also~\cite{bennett1999entanglement},~\cite{bennett2014quantum}, and~\cite{berta2011quantum}). While Shannon’s channel coding theorem simulates a noiseless channel from a noisy channel, the reverse Shannon theorem does the opposite — simulating a noisy channel from a noiseless channel. Hence, given unlimited common randomness, any memoryless channel can simulate any other channel of lower capacity. 	Cuff~\cite{cuff2013distributed,4595216} also considered the problem of channel simulation, but where the common randomness is a limited resource at rate $R_0$ and a clean channel of rate $R$ from Alice to Bob is given. He fully characterized the rate pairs $(R,R_0)$ for which Alice and Bob can simulate a channel that is arbitrarily close in total variation to a given memoryless channel. Haddadpour {\em et al}~\cite{haddadpour2017simulation} studied a similar problem, but where the channel from Alice to Bob is a noisy memoryless one, instead of a bit pipe. Other extensions to this problem can be found in the literature (e.g~\cite{gohari2011generating},~\cite{yassaee2015channel},~\cite{haddadpour2012coordination}, and~\cite{satpathy2016secure}). One can also consider other notions of channel simulation, e.g., average distortion measure as in\cite{neuhoff1979channels},~\cite{neuhoff1982channel}, and~\cite{steinberg1994channel}, agreement probability as in~\cite{bogdanov2011extracting} or the exact simulation of~\cite{kumar2014exact}. Ghazi {\em et al}~\cite{DBLP:journals/corr/GhaziKS16} and De {\em et al}~\cite{de2018non} studied the computational-theoretic problem of deciding whether certain distributions can be simulated from a given sequence of i.i.d. pairs in a setup similar to ours (but where the target distributions are more general), and gave conditions for decidability. Some impossibility results for our setup can be readily obtained from various forms of data processing inequalities. Clearly, a necessary condition for source simulation is that $I(X;Y) \leq I(U;V)$. Witsenhausen's results ~\cite{witsenhausen1975sequences} imply the maximal correlation necessary condition of $\rho _m(X;Y) \leq \rho _m(U;V)$. Recently, Kamath and Anantharam~\cite{kamath2016non} showed that source simulation is possible only if $\mathcal{R}(X;Y) \subseteq \mathcal{R}(U;V)$ where $\mathcal{R}$ is the hypercontractivity ribbon, which is the set of all pairs $(p,q)$ for which $(X,Y)$ is $(p,q)$-hypercontractive. 

\subsection*{Organization} 
The rest of the paper is organized as follows. In Section~\ref{sec:Prelim} notations and necessary mathematical background are provided. Inspired by a variant of the soft-covering lemma, we formulate a proof for Theorem \ref{thrm:achievable_region} and Proposition~\ref{prop:better_than_cgk} in Section \ref{sect:Sufficient Conditions}. Section \ref{sect:Convecture} is dedicated to the proof of Theorem~\ref{thrm:DSBS}, which is accomplished progressively by analyzing steps of increasing complexity. Summary and discussion appear in section~\ref{sect:Sum}.

\section{Preliminaries}\label{sec:Prelim}
\subsection{General Background and Notation}
Random variables (r.v.s) are denoted by upper-case letters, their realizations by corresponding lower-case letters. We use $p_X$ to denote the probability distribution of a random variable $X$ with alphabet $\mathcal{X}$, and we write $X \sim p_X$. We denote the i.i.d. distribution on $\mathcal{X}^n$ as $p_{X}^\tn$. The Shannon entropy of $X \sim p_X$ is defined as
\begin{align}
H(X)\triangleq -\sum_{x\in \mathcal{X}}p_X(x)\log p_X(x).
\end{align}
It is clear that for $X^n \sim p_{X}^\tn$, $H(X^n)=nH(X)$.
The set of probability measures on $\mathcal{X}$ is denoted as $\mathcal{P}(\mathcal{X})$.
The total variation ($\dtv$) distance between two probability mass functions $p_X$ and $q_X$ with a common alphabet $\mathcal{X}$ is defined as
\begin{align}
\dtv(p_X,q_X)\triangleq \underset{A \subseteq \mathcal{X}}{ \sup }|p_X(A)-q_X(A)| \label{eq:d_TV}.
\end{align}
The following Lemma is standard.
\begin{lemma} \label{lem:TV_from_const}
The total variation distance satisfies
\begin{align}
    \dtv(p_X,q_X)=\frac{1}{2}\sum_{x\in \mathcal{X}}|p_X(x)-q_X(x)|. 
\end{align}
Specifically, $\dtv(p_X,\ind(x=\alpha)) = \Pr(X\neq \alpha)$ where $X\sim p_X$. 
\end{lemma}

The Kullback Leibler (KL) Divergence between $p_X$ and $q_X$ is defined as
\begin{align}
D(p_X\|q_X)\triangleq \sum_{x\in \mathcal{X}}p_X(x)\log \frac{p_X(x)}{q_X(x)} .
\end{align}
Now consider another r.v $Y \sim p_Y$. Then the mutual information between $X$ and $Y$ is defined as
\begin{align}
I(X;Y)\triangleq D(p_{XY}\|p_Xp_Y)= \sum_{x\in \mathcal{X},y\in \mathcal{Y}}p_{XY}(x,y)\log \frac{p_{X|Y}(x|y)}{p_X(x)}.
\end{align}
The following lemma shows that we can break down the divergence between some general distribution and an i.i.d. distribution into two nonnegative quantities: one that captures the deviation of the sequence from being i.i.d., and the other that captures the deviation of the marginals from the target marginal. Any upper bound on the divergence will therefore also upper bound each of these two quantities. 

\begin{lemma} \label{lem:Divbound2}
	It holds that 
	\begin{align}
		D\left( p_{X^nY^n}\| p^\tn_{XY}\right) =  \sum_{i=1}^{n}I\left(X_i,Y_i;X^{i-1},Y^{i-1}\right) + \sum_{i=1}^{n}D\left( p_{X_iY_i}\| p_{XY}\right). 
	\end{align}
\end{lemma}
 \begin{proof}
 	Write 
 	\begin{align} 
 	D\left( p_{X^nY^n}\| p^\tn_{XY}\right)&=\sum_{x^n,y^n}p_{X^nY^n}(x^n,y^n)\log\frac{p_{X^n,Y^n}(x^n,y^n)}{p^\tn_{XY}(x^n,y^n)}\\
 	&=-H(X^n,Y^n)-\sum_{x^n,y^n}p_{X^nY^n}(x^n,y^n)\log p^\tn_{XY}(x^n,y^n)\\
 	&=-H(X^n,Y^n)-\sum_{x^n,y^n}p_{X^nY^n}(x^n,y^n)\sum_{i=1}^{n}\log p_{XY}(x_i,y_i)\\
 	&=-H(X^n,Y^n)-\sum_{i=1}^{n}\sum_{x,y}\left(\sum_{\underset{(x_i,y_i)=(x,y)}{x^n,y^n}}p_{X^nY^n}(x^n,y^n)\right)\log p_{XY}(x,y)\\
 	&=-H(X^n,Y^n)-\sum_{i=1}^{n}\sum_{x,y}p_{X_iY_i}(x,y)\log p_{XY}(x,y)\\
 	&=-H(X^n,Y^n)+\sum_{i=1}^{n}\left(H(X_i,Y_i)+D\left( p_{X_iY_i}\parallel p_{XY}\right)\right)\\
 	&=\sum_{i=1}^{n}H(X_i,Y_i)-H(X^n,Y^n)+\sum_{i=1}^{n}D\left( p_{X_iY_i}\parallel p_{XY}\right)\\
 	&= \sum_{i=1}^{n}I\left(X_i,Y_i;X^{i-1},Y^{i-1}\right) + \sum_{i=1}^{n}D\left( p_{X_iY_i}\| p_{XY}\right), 
 	\end{align}  
 	as desired.
 \end{proof}

The correlation between $X$ and $Y$ is defined as
\begin{align}
\rho(X,Y) = \frac{\Cov (X,Y)}{\sqrt{\Var(X)}\sqrt{\Var(Y)}},
\end{align}
while the {\em Hirschfled-Gebelein-R{\`e}nyi maximal correlation} between $X$ and $Y$ (which we will refer to simply is maximal correlation) was defined in~\cite{witsenhausen1975sequences} as: 
\begin{align}
\rho_m (X;Y) = \underset{f,g}{\sup}\rho \left(f(X),g(Y)\right) .
\end{align}
We introduce the following well-known lemma without proof.
\begin{lemma}
The maximal correlation holds the following properties:
\begin{enumerate}
	\item (data processing inequality) For any functions $f,g$, $\rho_m (X;Y) \geq \rho_m (f(X);g(Y))$.
	\item (tensorization~\cite{witsenhausen1975sequences}) If $(X_1,Y_1), (X_2,Y_2)$ are independent, then
	\begin{align}
	\rho_m (X_1,X_2;Y_1,Y_2)=\max \{\rho_m (X_1;Y_1),\rho_m (X_2;Y_2)\}.  
	\end{align}
\end{enumerate}
Furthermore, $X$ and $Y$ are independent if and only if they have zero maximal correlation~\cite{renyi1959measures}. 
\end{lemma}
\subsection{Boolean Functions and Fourier Analysis} \label{sec:MB_Boolean}
Any real-valued function $f:\{-1,1\}^n \rightarrow \mathbb{R}$ on the Hamming cube can be uniquely expressed as a multilinear polynomial~\cite{o2014analysis}
\begin{align}
f(u^n)=\sum_{S \subseteq [n]}^{}\hat{f}_Su^S, \label{eq:line_parity}
\end{align} 
where $u^S = \prod_{i\in S}u_i$. This is known as the {\em Fourier expansion} of $f$, and the real numbers $\hat{f}_S$ are called the {\em Fourier coefficients} of $f$. Collectively, the coefficients are called the Fourier spectrum of $f$. This simple representation will encourage us to transform our state space from $\{0,1\}^n$ to $\{-1,1\}^n$. We define an inner product $\langle \cdot,\cdot \rangle$ on pairs of functions $f,g$ by:
\begin{align}
\langle f,g\rangle  = \E[f(U^n)g(U^n)],
\end{align}
where it is assumed that $U^n$ is distributed uniformly over $\{-1,1\}^n$. Hence the norm of a Boolean function is
\begin{align}
\lVert f \rVert_2^2=\langle f,f \rangle  = \E[f^2(U^n)] .
\end{align}
It is readily observed that the number $u^S$ is a Boolean function; it computes the logical \textit{parity}, or \textit{exclusive-or} (XOR), of the bits $(u_i)_{i \in S}$.~\eqref{eq:line_parity} then means that any $f$ can be represented as a linear combination of parity functions over the reals. Moreover, the $2^n$ parity functions form an orthonormal basis for the vector space $V$ of functions $\{-1,1\}^n \rightarrow \mathbb{R}$, i.e 
\[\langle U^S,U^T\rangle= \left\{ \begin{array}{ll}
1 & if \hspace{2mm} S=T,\\
0 & S \neq T.\end{array} \right. \]
This follows since $u^S u^T=u^{S\Delta T}$, where $S\Delta T$ denotes symmetric
difference, and 
\[\E[U^s]=\E\left[\prod_{i\in S}^{}U_i\right]= \left\{ \begin{array}{ll}
1 & if \hspace{2mm} S=\emptyset,\\
0 & S \neq\emptyset.\end{array} \right. \]
Hence, The Fourier expansion of $f:\{-1,1\} \rightarrow \mathbb{R}$ is essentialy the representation of $f$ over the orthonormal basis of parity functions $\left(U^S\right)_{S \subseteq [n]}$, equivalently $\left\langle f,U^S\right\rangle=\hat{f}_S$.
The orthonormal basis of parities also allows us to measure the norm of $f:\{-1,1\}^n \rightarrow \mathbb{R}$ efficiently: It is just the sum of the squares of $f$'s Fourier coefficients, a fact known as \textit{Parseval's Theorem}.
\begin{align}
\langle f,f \rangle  = \E[f^2(U^n)]=\sum_{S \subseteq [n]}^{}\hat{f}_S^2.
\end{align} 
More generally, given two functions $f,g:\{-1,1\}^n \rightarrow \mathbb{R}$ we can compute their inner product by taking the "dot product" of their corresponding Fourier coefficients, which is known as  \textit{Plancherel's Theorem}.
\begin{align}
\langle f,g \rangle  =	\left\langle \sum_{S \subseteq [n]}\hat{f}_SU^S,\sum_{T \subseteq [n]}\hat{g}_TU^T \right\rangle =\sum_{S,T \subseteq [n]}\hat{f}_S\hat{g}_S\left\langle U^S,U^T \right\rangle=\sum_{S \subseteq [n]}\hat{f}_S\hat{g}_S.
\end{align}
$f$ is called Boolean if $f:\{-1,1\}^n\to\{-1,1\}$. In this case, note that
\begin{align}
\sum_{S \subseteq [n]}^{}\hat{f}_S^2=1.
\end{align} 
The expected value of a Boolean function $f$ can be calculated either by
\begin{align}
\E[f]=\Pr (f=1)-\Pr (f=-1)=2\Pr (f=1)-1,
\end{align}
or directly from the Fourier coefficients, since $\E[f]=\langle f,1 \rangle = \hat{f}_\emptyset$.
The {\em Fourier weight of f at degree k} is defined as
\begin{align}
W^k[f]=\sum_{|S|=k}\hat{f}_S^2.
\end{align}
Note that $\sum_{k=0}^n W^k[f] =\E\left[f^2(U^n)\right]$ and if $f$ is Boolean then $\sum_{k=0}^n W^k[f] = 1$. 
For later reference, we prove the following lemma:  
\begin{lemma}\label{lem:corrbound}
Let $(U^n,V^n) \sim p_{UV}^\tn$ where $p_{UV}$ is a $\dsbs(p)$. Then
\begin{align}
\E[f(U^n)g(V^n)]&=\sum_{S\subseteq [n]}(1-2p)^{|S|}\hat{f}_S\hat{g}_S\\
&\leq \frac{1}{2}\sum_{k=0}^n(1-2p)^k(W^k[f]+W^k[g]).
\end{align} 
\end{lemma}
\begin{proof}
	It holds that $V_i=U_iZ_i$ where $Z_i \sim \bern(p)$ are i.i.d. r.vs.
	Write 	
	\begin{align}
	\E[f(U^n)g(V^n)]&=\sum_{S,T\subseteq [n]}\hat{f}_S\hat{g}_T\E[U^{S \Delta T}]\E[Z^T]\\&=\sum_{S\subseteq [n]}\hat{f}_S\hat{g}_S(1-2p)^{|S|}=\sum_{k=1}^{n}\sum_{|S|=k}(1-2p)^{k}\hat{f}_S\hat{g}_S.
	\end{align}
	Applying the Cauchy-Schwartz and the arithmetic-geometric mean inequalities, we have: 
	\begin{align}
	&\forall k, \sum_{|S|=k}\hat{f}_S\hat{g}_S\leq \sqrt{\sum_{|S|=k}\hat{f}_S^2\sum_{|S|=k}\hat{g}_S^2}=\sqrt{W^k[f]W^k[g]}, \\
	&\sum_{k=1}^{n}(1-2p)^{k}\sqrt{W^k[f]W^k[g]} \leq  \sum_{k=1}^{n}(1-2p)^{k}\frac{W^k[f]+W^k[g]}{2},
	\end{align}
	and the result follows.	
\end{proof}
\begin{definition}
If $f(u^n)=u_i$ for some $i \in [n]$, then $f$ is called a dictator function.
\end{definition}
Another result that will provide important insight for us going forward is the Friedgut-Kalai-Naor Theorem \cite{friedgut2002boolean}, which states that if the Fourier coefficients are concentrated on the first level, then the function is close to a dictator function:
\begin{lemma} \label{lem:FKN}
(FKN \cite{friedgut2002boolean}) If $f$ is a Boolean function and
\begin{align}
    \sum_{k\neq 1}W^k[f]\leq \delta
\end{align}
then there is a $j \in [n]$ and $b \in \{-1,1\}$ such that 
\begin{align}
    	\Pr(f(U^n)\neq b\cdot U_j) &\leq K\cdot \delta
\end{align}
for some constant K.
\end{lemma}
	
\section{Simulable Distributions - Achievable Region} \label{sect:Sufficient Conditions}
The main tool used in our proof of Theorem~\ref{thrm:achievable_region} is the so-called \textit{soft-covering lemma}, which has its origins in Wyner's work~\cite{wyner1975common}. 
\begin{lemma}[Lemma VII.9 in ~\cite{cuff2013distributed}]\label{lem:soft}
		Let $p_{XUW}$ be given. If $H(U)>I(X;U,W)$, then there exists a sequence of encodings $a_n: \mathcal{U}^n \rightarrow \mathcal{W}^n$ such that if $U^n\sim p_U^\tn$ and $X^n\sim p_{X|UW}^\tn(\cdot \mid  U^n,a_n(U^n))$, then $p_{X^n}$ converges in relative entropy to an i.i.d. distribution with marginal $p_X$, i.e., 
		\begin{align}
		\lim_{n\to\infty}D\left(p_{X^n}\parallel p_X^\tn\right) = 0.
		\end{align} 
	\end{lemma}
	\begin{center}
		\begin{figure}[H]
			\centering
			\begin{tikzpicture}
			\node (u) {$U^n$};
			\node (encu) [rectangle, rounded corners, draw, fill=white, drop shadow, anchor=west, text=black, minimum size=12mm] at ($(u.east)+(6mm,0mm)$)  {Enc};
			\node (chan) [rectangle, rounded corners, draw, fill=white, drop shadow, anchor=west, text=black, minimum size=20mm] at ($(encu.east)+(12mm,-4mm)$) {$p_{X|W,U}$};
			\node (x) at ($(chan.east)+(8mm,0mm)$)  {$X^n$};
			\path (u) edge[thick, ->] (encu);
			\draw[thick,->] (0,-0.2) -- (0,-1.1)-- (3.2,-1.1);
			\draw[thick,->] (2.2,0) -- (3.2,0) node [above left] {$W^n$};
			\path (chan) edge[thick, ->] (x);
			\end{tikzpicture}
			\caption{Lemma~\ref{lem:soft} - Soft covering}
			\label{fig:Soft_Covering}
		\end{figure}
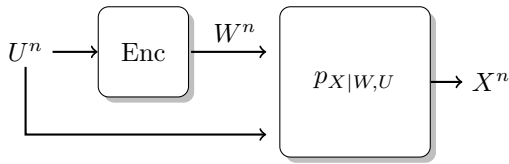
	\end{center}
This lemma was proved by Cuff~\cite{cuff2013distributed} for a weaker convergence in total variation. He showed that
	\begin{align}
	\E \dtv \left(p_{X^n},p_X^\tn\right)\leq \frac{3}{2}\exp(-\gamma n)
	\end{align} 
where $\gamma$ is some positive constant, hence there exists a codebook for which the total variation is exponentially decaying. However, as noted in~\cite{cuff2013distributed}, an inequality from~\cite{cover2012elements}
can be used to show that in this case convergence in total variation also implies convergence in KL divergence. Specifically, the inequality states that if $\Pi$ is absolutely continues with respect to $\Gamma$ and $\Gamma$ is an i.i.d. discrete distribution, then
	\begin{align}
	D\left(\Pi\| \Gamma\right) \in O\left(\left(n+\log\frac{1}{\dtv \left(\Pi,\Gamma\right)}\right)\dtv \left(\Pi,\Gamma\right)\right).
	\end{align}	
This implies that in our setup, the KL divergence is controlled by the total variation distance, and the soft-covering lemma follows. Nevertheless, we provide an alternative proof of Lemma~\ref{lem:soft}, based on ideas from~\cite{cuff2013distributed} and~\cite{cuff2015stronger}, that works with the KL divergence directly in Appendix~\ref{app:soft}. 

\subsection{Proof of Theorem~\ref{thrm:achievable_region}}	
Our construction is based on {\em hybrid coding} in the spirit of~\cite{minero2015unified},~\cite{soundararajan2009hybrid}. We use the GK common part as the digital part, and $U^n$ (resp. $V^n$) as the analog (scalar) part. Alice and Bob both remotely compute the GK common part $K^n$ of $(U^n,V^n)$ from their respective components, and create $W^n = a(K^n)$ using some encoding $a:\mathcal{K}^n\to\mathcal{W}^n$. Alice then generates $X^n\sim p_{X|UW}^\tn(\cdot\mid U^n, a(K^n))$ and Bob generates $Y^n\sim p_{Y|VW}(\cdot\mid V^n,a(K^n))$ using local randomness. This setup is depicted in Figure~\ref{fig:DSS_achievability}.
\begin{center}
	\begin{figure}[H]
	\centering
	\resizebox*{0.5\textwidth}{!}{
		\begin{tikzpicture}[node distance=0.1cm, auto]  
		\tikzset{
			mynode1/.style={rectangle,draw=black, top color=white, bottom color=white, thick, minimum width=0.8cm, minimum height = 0.8cm, text centered},
			mynode2/.style={rectangle,draw=black, top color=white, bottom color=white, thick, inner ysep=2.5em, text centered},
			myarrow/.style={->, >=latex',shorten >=1pt, line width=0.3mm},
			mylabel/.style={text width=5em, text centered}  
		} 
		\draw (4,2.6) node (Alice) {\Large{Alice} };
		\draw (0.8,1.7) node (U) {$U^n$};
		\draw (9,1) node (X) {$X^n$};
		\draw (2.9,1.9) node (K) {$K^n$};
		\draw (5.1,1.9) node (W) {$W^n$};
		\draw (4,1.7) node [mynode1,text width=3em] (encoder) {Enc};
		\draw (7,1) node [mynode2,text width=4em] (channel) {$P_{X|U,W}^\tn$};
		\draw (2.2,1.7) node [mynode1,text width=1em] (fun) {f};
		
		\draw[myarrow] (U.east) |- (fun.west);
		\draw[myarrow] (fun.east) |- (encoder.west);
		\draw[myarrow] (U.east) |- (6,0);
		\draw[myarrow] (encoder.east) |- (6,1.7) ;
		\draw[myarrow] (channel.east) |- (X.west);		
		\end{tikzpicture} 
	}
		\caption{Theorem~\ref{thrm:achievable_region} - Hybrid coding}
		\label{fig:DSS_achievability}
\end{figure}
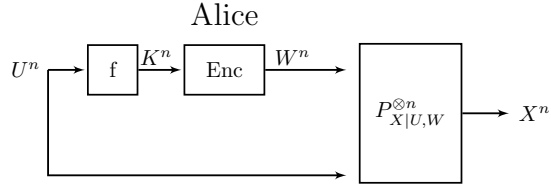
\end{center} 
Let us define $\widetilde{U}^n = (U^n,V^n)$ and $\widetilde{X}^n = (X^n,Y^n)$. Note that we cannot use Lemma~\ref{lem:soft} directly on $\widetilde{X}^n,\widetilde{U}^n, W^n$, since $W^n$ is generated from $K^n$ and not from the entire $\widetilde{U}^n$. Instead, we show that $\widetilde{X}^n$ is generated from $K^n$ and $W^n = a(K^n)$ in a memoryless fashion via $p_{\widetilde{X}|KW}^\tn$. 
\begin{align}
	p_{\widetilde{X}^n}(\widetilde{x}^n)&=\sum_{\widetilde{u}^n,k^n} p_{\widetilde{U}K}^\tn(\widetilde{u}^n,k^n)p_{\widetilde{X}|\widetilde{U}W}^\tn(\widetilde{x}^n|\widetilde{u}^n,a(k^n))\\
	&=\sum_{k^n}p_K^\tn(k^n)\sum_{\widetilde{u}^n} p_{\widetilde{U}|K}^\tn(\widetilde{u}^n|k^n)p_{\widetilde{X}|\widetilde{U}W}^\tn(\widetilde{x}^n|\widetilde{u}^n,a(k^n))\\
	&=\sum_{k^n}p_K^\tn(k^n)\sum_{\widetilde{u}^n} p_{\widetilde{U}|KW}^\tn(\widetilde{u}^n|k^n,a(k^n))p_{\widetilde{X}|\widetilde{U}WK}^\tn(\widetilde{x}^n|\widetilde{u}^n,a(k^n),k^n)\\
	&=\sum_{k^n}p_{K^n}(k^n)\sum_{\widetilde{u}^n} p_{\widetilde{X}\widetilde{U}|K,W}^\tn(\widetilde{x}^n,\widetilde{u}^n|k^n,a(k^n))\\
	&=\sum_{k^n}p_{K}^\tn(k^n)p_{\widetilde{X}|KW}^\tn(\widetilde{x}^n|k^n,a(k^n)), 
\end{align}
where we have used the fact that $\widetilde{U}-K-W$ and $\widetilde{X}-(\widetilde{U},W)-K$ are Marokv chains. Applying Lemma~\ref{lem:soft} with $(X,U,W) \leftarrow (\widetilde{X},K,W)$, we find that if $H(K)>I(\widetilde{X};K,W)=I(X,Y;K,W)$, then there exist encodings such that the statement of the theorem holds.
	
\subsection{Proof of Proposition~\ref{prop:better_than_cgk}}
We first show that the digital solution is covered by Theorem 2, i.e., that $\cgk(U;V) \geq \cw(X;Y)$ is a sufficient condition for $p_{XY}$ to be simulable from $p_{UV}$. To show that, let us choose $p(w|k)$ such that the Markov chain~\eqref{eq:markov1} is satisfied. Let us further impose $X-W-(U,V)$ and $Y-W-(X,U,V)$, which imply the Markov chain~\eqref{eq:markov2}. The theorem then indicates that $p_{XY}$ is simulable if $\cgk(U;V)\geq I(X,Y; W)$. We can now minimize over all suitable $W$ to obtain the sufficient condition.
To show that the analog solution is covered by the theorem, let us choose $W$ to be independent of $(U,V,X,Y)$. The Markov chain~\eqref{eq:markov1} is satisfied and also $I(X,Y; K,W) = I(X,Y;K) \leq H(K) = \cgk(U;V)$ holds. The only additional condition is the Markov chain~\eqref{eq:markov2}, which in this case reduces to $X-U-V-Y$. 

To show that the inclusion is strict for some $p_{UV}$, let $p_U(0) = p_U(1) = 0.4$, $p_U(2) = p_U(3) = 0.1$, $p = 0.1$, and $p_{V|U}$ is given in Figure~\ref{fig:pUV_new_region} as follows.
\begin{center}
	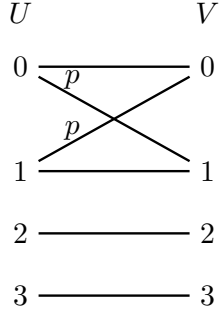
\begin{figure}[H]
		\centering
		\resizebox*{0.22\textwidth}{!} {
		\begin{tikzpicture}
		\node (u0) {$0$};
		\node (u) at ($(u0.north)+(0mm,4mm)$){$U$};
		\node (u1) at ($(u0.south)+(0mm,-10mm)$) {$1$};
		\node (u2) at ($(u1.south)+(0mm,-5mm)$) {$2$};
		\node (u3) at ($(u2.south)+(0mm,-5mm)$) {$3$};
		\node (v0) at ($(u0.east)+(20mm,0mm)$) {$0$};
		\node (v) at ($(v0.north)+(0mm,4mm)$){$V$};
		\node (v1) at ($(v0.south)+(0mm,-10mm)$) {$1$};
		\node (v2) at ($(v1.south)+(0mm,-5mm)$) {$2$};
		\node (v3) at ($(v2.south)+(0mm,-5mm)$) {$3$};
		\node (p1) at ($(u0.east)+(4mm,-1.5mm)$) {$p$};
		\node (p2) at ($(u0.east)+(4mm,-7.5mm)$) {$p$};
		\path (u0) edge[thick, -] (v0);
		\path (u1) edge[thick, -] (v1);
		\path (u2) edge[thick, -] (v2);
		\path (u3) edge[thick, -] (v3);
		\path (u0) edge[thick, -] (v1);
		\path (u1) edge[thick, -] (v0);
		\end{tikzpicture}
	}
		\caption{An illustration of a $p_{UV}$ for which  $\mathcal{S}_{\mathrm{dig}}(p_{UV}) \cup \mathcal{S}_{\mathrm{ana}}(p_{UV})\subset \mathcal{S}(p_{UV})$}
		\label{fig:pUV_new_region}
	\end{figure}
\end{center}
Now let $(X,Y) \sim \dsbs(q)$ where $q\geq 0$. Using strictly digital scheme (i.e Wyner coding) we can approximately simulate all $\dsbs(q)$ such that
\begin{align}
     &H(K)=\cgk(U;V)\geq \cw(X;Y)\\  &H(0.1,0.1,0.8)\geq 1+h_b(q)-2h_b\left(\frac{1}{2}-\frac{1}{2}\sqrt{1-2q}\right),
\end{align}
where we have used the expression for the Wyner common information of a $\dsbs(q)$. For the given $p_K$, we can simulate all $q\geq 0.065$. Now, the scalar scheme that achieves the lowest possible $q$ is the following: If $U \in \{0,2\}$ then $X=0$, otherwise $X=1$. In a similar way, If $V \in \{0,2\}$ then $Y=0$, otherwise $Y=1$. In this case we can exactly simulate $\dsbs\left(0.8\cdot p\right)$. Now consider the following hybrid scheme: Alice and Bob generate a codebook from $K^n$ to $W^n$ that achieves the minimum in Wyner's common information. The simulation protocol is the following: If $U_i=V_i=2$, Alice and Bob output $0$, and if $U_i=V_i=3$, Alice and Bob output $1$. However, if $(U_i,V_i) \in \{0,1\}$, Alice and Bob use $W_i$ from the encoding of $K^n$ and pass it through the correct $p_{X|W}$ (resp. $p_{Y|W}$) in the Wyner scheme. In this way, they can approximately simulate a $\dsbs(0.8 \cdot 0.065)$, better than both schemes.

We now prove the second statement of Proposition~\ref{prop:better_than_cgk}. Assume first that $\cgk(U;V)=0$. Then clearly $W$ is independent of $(U,V)$ and also independent of $(X,Y)$. Now fix any $w_0\in \mathcal{W}$, and write 
	\begin{align}
		p_{XY}(x,y) &= p_{XY|W}(x,y|w_0)\\
			&=\sum_{u,v}p_{UV}(u,v)p_{XY|UVW}(x,y|u,v,w_0)\\
			&= \sum_{u,v}p_{UV}(u,v)p_{X|UW}(x|u,w_0)p_{Y|VW}(y|v,w_0)
	\end{align}	  
Hence, considering the r.v.s $(\tilde{X},\tilde{Y})$ generated via 
		\begin{align}
			&\tilde{p}_{\tilde X|U}(\tilde x|u) \dfn  p_{X|U,W}(\tilde x|u,w_0)
			\\&\tilde{p}_{\tilde Y|V}(\tilde y|v)  \dfn p_{Y|V,W}(\tilde y|v,w_0), 
	\end{align}	
we have that $\tilde{X}-U-V-\tilde{Y}$ forms a Markov chain, and also $(\tilde{X},\tilde{Y})\sim p_{XY}$. 
		
Conversely, suppose $\cgk(U;V) = \epsilon > 0$. Consider the set of simulable distributions generated by some scalar Markov chain $X-U-V-Y$. Each of these distributions can be written in matrix form as
	\begin{align}
	\mathbf{P}_{XY}=\mathbf{P}_{X|U}\mathbf{P}_{UV}\mathbf{P}_{Y|V}^T, 
	\end{align}	
	hence in particular, recalling that $\rank(AB) \leq \min(\rank(A),\rank(B))$, it must hold that 
	\begin{align}
		    \rank(\mathbf{P}_{XY}) \leq \rank(\mathbf{P}_{UV}).
	\end{align}
Now, appealing to the digital approach, it suffices to show that there exists a Markov chain $X-W-Y$ such that $\rank(\mathbf{P}_{XY}) \geq \rank(\mathbf{P}_{UV})$. To that end, choose $W$ to have support over an alphabet of cardinality $M > \rank(\mathbf{P}_{UV})$, and let $|\mathcal{X}| = |\mathcal{Y}| = M$ as well. The Markov structure implies that \begin{align}
		 \mathbf{P}_{XY}=\mathbf{P}_{X|W}\mathbf{P}_{W}\mathbf{P}_{Y|W}^T.
\end{align}
Since $\rank(\mathbf{P}_W) = M$ by construction, it suffices to show one can choose $\mathbf{P}_{X|W}$ and $\mathbf{P}_{Y|W}$ to have full rank, while keeping $I(X,Y;W) = I(X;W) + I(Y;W) \leq \epsilon$. This is an easy consequence of the fact that mutual information is continuous w.r.t. the $L^\infty$ metric, whereas matrix rank is not. In particular, pick some small $\alpha>0$ and fix any  column probability vector $v\in \mathbbm{R}^M$ with all entries are in the $(2\alpha, 1-2\alpha)$ range. Let $A\in \mathbbm{R}^{M\times M}$ be a matrix whose columns are all equal to $v$. Now, pick $P_{X|W}$ at random inside an $L^\infty$ ball of radius $\frac{\alpha}{M}$ around $A$ within the space of conditional probability matrices, w.r.t. the Lebesuge measure restricted to that space. Since the volume of subspaces of dimension smaller than $M$ is zero, then $\Pr(\rank(P_{X|W}) = M)=1$. On the other hand, noting that $P_{X|W} = A$ yields $I(X;W)=0$, the continuity of the mutual information w.r.t. the $L^1$ metric~\cite{csiszar2011information} along with the fact that $||\cdot ||_1 \leq M||\cdot ||_\infty$ implies that $\Pr(I(X;W) < \epsilon/2) = 1$, if $\alpha>0$ is taken to be small enough. Hence, there exists a specific $P_{X|W}$ satisfying both the rank and the mutual information requirements. A similar argument can be made for $P_{Y|W}$, concluding the proof.

\section{The Binary Symmetric Case}\label{sect:Convecture}
Loosely speaking, Theorem~\ref{thrm:DSBS} says that if we are able to reliably distributively simulate a $\dsbs(p+\delta)$ from a $\dsbs(p)$, and if $\delta$ is sufficiently small, then the channels from $U^n$ to $X^n$ and from $V^n$ to $Y^n$ are close to being scalar and memoryless. We find it instructive to prove this claim in steps: In section~\ref{subsec:exact_sim} we show that to simulate a $\dsbs(p)$ exactly, $X^n$ and $Y^n$ must necessarily be the same signed coordinate permutation of $U^n$ and $V^n$.  In section~\ref{subsec:almost_exact_sim} we show that to simulate a $\dsbs(p)$ with some KL divergence of at most $\epsilon$ via deterministic scheme, the mappings $f(u^n), g(v^n)$ must be almost equal to the same signed coordinate permutation.  In section~\ref{subsec:almost_exact_sim_rand} we show that to simulate a $\dsbs(p)$ with some KL divergence of at most $\epsilon$ via randomized scheme, \textit{with high probability} the mappings must be almost equal to the same signed coordinate permutation. Finally, In section~\ref{subsec:theorem2} we prove Theorem~\ref{thrm:DSBS}.

\subsection{Deterministic Schemes} 
We begin by limiting our discussion to deterministic simulation schemes, i.e., where $X^n$ (resp. $Y^n$) is a deterministic function of $U^n$ (resp. $V^n$). 

\subsubsection{Exact Simulation ($\epsilon=0, \delta=0$)}\label{subsec:exact_sim}
In this subsection, we prove Theorem~\ref{thrm:DSBS} for the case where both $\epsilon=0$ and $\delta=0$, i.e., where $(X^n,Y^n)$ is a memoryless $\dsbs(p)$. Clearly, one way to guarantee this is to generate $X^n$ and $Y^n$ from $U^n$ and $V^n$ respectively via the same {\em signed} coordinate permutation, i.e., a coordinate permutation that possibly flips some of the coordinates as well. Theorem~\ref{thrm:DSBS} indicates that this is the {\em only} way to do this. 

Let $X^n=f(U^n)$ and $Y^n=g(V^n)$ be such that $p_{X^nY^n} = p_{XY}^\tn$, where $p_{XY}$ is a $\dsbs(p)$. In particular, $X^n$ (resp. $Y^n$) are uniformly distributed over the entire Hamming cube, hence it is clear that $f$ and $g$ must be permutations of the Hamming cube. Thus on the one hand, by assumption, we have that 
\begin{align}
	p_{X^nY^n}(x^n,y^n) &=  2^{-n}p^{d_H(x^n,y^n)}(1-p)^{n-d_H(x^n,y^n)}
\end{align}
and on the other hand 
\begin{align}
p_{X^nY^n}(x^n,y^n) &=  \Pr\left(f(U^n)=x^n,g(V^n)=y^n\right)\\
&=\Pr\left(U^n=f^{-1}(x^n),V^n=g^{-1}(y^n)\right)\\
&=2^{-n}p^{d_H(f^{-1}(x^n),g^{-1}(y^n))}(1-p)^{n-d_H(f^{-1}(x^n),g^{-1}(y^n))}.
\end{align}
hence it must be that 
\begin{align}
	d_H(x^n,y^n) = d_H(f(x^n), g(y^n))
\end{align}
for any $x^n,y^n \in\{0,1\}^n$. Substituting $y^n=x^n$ in the above, we see that $d_H(f(x^n), g(x^n)) = 0$ for any $x^n$. We thus conclude that $f=g$ must hold. The problem is now reduced to establishing the following Lemma. 
\begin{lemma}
A bijection $f:\{0,1\}^n\to \{0,1\}^n$ preserves the Hamming distance if and only if $f$ is a signed coordinate permutation. 
\end{lemma}
\begin{proof}
This is a well known fact, see e.g.~\cite{frucht1949groups}, but we nevertheless provide a short proof. A signed coordinate permutation is clearly a bijection that preserves the Hamming distance. To prove the other direction, assume first that $f(0^n) = 0^n$. Then it must be that $f$ preserves the Hamming weight, and specifically, it permutes the vectors of weight one, hence it must be a coordinate permutation. The case where $f(0^n)$ is mapped to any other nonzero vector is similar, with the exception that $f$ is now a signed coordinate permutation, flipping exactly those coordinates where $f(0^n)$ is one. 
\end{proof}

\subsubsection{Almost Exact Simulation ($\epsilon>0, \delta =0$)}\label{subsec:almost_exact_sim}
We saw that the only way to simulate a $\dsbs(p)$ from  a $\dsbs(p)$ is the trivial way, by signed coordinate permutations. Next, we examine the stability of this claim. Namely, we allow the simulation to be slightly imperfect, such that that KL divergence between the simulated distribution and a $\dsbs(p)$ is at most $\epsilon$, and show that both the functions $f(u^n)$ and $g(v^n)$ will be almost equal to the same signed coordinate permutation.      

First, although we do not directly use this fact, it is instructive to  note that both $f$ and $g$ are almost permutations of the Hamming cube.
\begin{lemma}\label{lem:mostly_bijection}
$\Pr\left(\max\left\{|f^{-1}(X^n)|, |g^{-1}(Y^n)|\right\} >  1\right) \leq 2\epsilon$. 
\end{lemma}
\begin{proof}
	Observe that by the chain rule, the marginal divergence must also be bounded by $\epsilon$, hence  
	\begin{align}	
	D\left( p_{X^n}\parallel p^\tn_{X}\right)  = n - H(X^n) \leq \epsilon, 
	\end{align}
	and therefore $H(X^n) \geq n -\epsilon$. Then 
	\begin{align}
	n &= H(U^n)\\
	&= H(U^n, X^n)\\
	&= H(X^n) + H(U^n | X^n)\\
	&\geq  H(X^n) + \Pr(|f^{-1}(X^n)|=1)\cdot 0 + \Pr(|f^{-1}(X^n)| > 1)\cdot 1,  
	\end{align}
	which implies that 
	$\Pr(|f^{-1}(X^n)| >  1) \leq \epsilon $. A similar arguments applies to $g$, and the claim follows from the union bound. 
\end{proof}

Next, we provide a useful lower bound on the KL divergence between the simulation $P_{X^nY^n}$ and the desired i.i.d. distribution $P_{XY}^n$, in terms of the expected Hamming distance only. 
\begin{lemma} \label{lem:Divbound}
Let $p_{UV}$ and $p_{XY}$ be $\dsbs(p)$ and $\dsbs(q)$ respectively. Let $X^n=f(U^n)$ and $Y^n = g(V^n)$. Then 
\begin{align}
D\left(p_{X^nY^n}\|p^\tn_{XY}\right) \geq \log{\frac{1-q}{q}}\cdot \E d_H(X^n,Y^n) +  n\left(\log{\frac{1-p}{1-q}} - p\log{\frac{1-p}{p}}\right).
\end{align}
with equality if and only if both $f$ and $g$ are bijections. Specifically, for $p=q$ we have 
\begin{align}
D\left(p_{X^nY^n}\|p^\tn_{XY}\right) \geq \log{\frac{1-p}{p}}\cdot \left(\E d_H(X^n,Y^n) -  np\right).
\end{align}
\end{lemma}
\begin{proof}
Write
\begin{align}	
D\left(p_{X^nY^n}\|p^\tn_{XY}\right)&=\sum_{x^n,y^n}p_{X^nY^n}(x^n,y^n)\log\frac{p_{X^nY^n}(x^n,y^n)}{p^\tn_{XY}(x^n,y^n)}\\
&=\sum_{x^n,y^n}\left(\sum_{u^n\in f^{-1}(x^n), v^n\in g^{-1}(y^n)}p_{UV}^\tn(u^n,v^n)\right)\log\frac{\left(\sum_{u^n\in f^{-1}(x^n), v^n\in g^{-1}(y^n)}p_{UV}^\tn(u^n,v^n)\right)}{p^\tn_{XY}(x^n,y^n)}\\
&=\sum_{u^n,v^n}p_{UV}^\tn(u^n,v^n)\log\frac{\left(\sum_{\tilde{u}^n\in f^{-1}(f(u^n)), \tilde{v}^n\in g^{-1}(g(v^n))}p_{UV}^\tn(\tilde{u}^n,\tilde{v}^n)\right)}{p^\tn_{XY}(f(u^n),g(v^n))}\\
&\geq \sum_{u^n,v^n}p_{UV}^\tn(u^n,v^n)\log\frac{p_{UV}^\tn(u^n,v^n)}{p^\tn_{XY}(f(u^n),g(v^n))}\\
&=\sum_{u^n,v^n}p_{UV}^\tn(u^n,v^n)\log\frac{p^{d_H(u^n,v^n)}(1-p)^{n-d_H(u^n,v^n)}}{q^{d_H(f(u^n),g(v^n))}(1-q)^{n-d_H(f(u^n),g(v^n))}}\\&=\log{\frac{1-q}{q}}\cdot \E d_H(f(U^n),g(V^n)) +  n\left(\log{\frac{1-p}{1-q}} - p\log{\frac{1-p}{p}}\right).
\end{align} 
It is easy to see that the inequality holds with equality if and only if both $f$ and $g$ are bijections. 
\end{proof}

Let us now write $X_i=f_i(U^n),Y_i=g_i(V^n)$, where $f_i,g_i$ are the Boolean functions generating the $i$th coordinate in the respective sequences. It is clear that for a successful simulation, these functions must be close to unbiased, i.e., $W^0[f_i]\approx 0, W^0[g_i]\approx 0$ . The following Lemma quantifies this fact. 
\begin{lemma}\label{lem:almost_unbiased}
	$\sum_{i\in [n]} W^0[f_i] + W^0[g_i] \leq \epsilon\cdot 4\ln{2}$. 
\end{lemma}
\begin{proof}
	Observe that by the chain rule, the marginal divergence must also be bounded by $\epsilon$, and so 
	\begin{align}	
	D\left( p_{X^n}\parallel p^\tn_{X}\right)  = n - H(X^n) \leq \epsilon. 
	\end{align}
Hence, $H(X^n) \geq n-\epsilon$. Note that $\hat{f}_{i,\emptyset} = 2\Pr(X_i=1)-1$ and hence by subadditivity of the entropy and Pinsker's inequality we have 
\begin{align}
	H(X^n) &\leq \sum_{i\in[n]} H(X_i) \\
	&= \sum_{i\in[n]} h_b(1/2 + \hat{f}_{i,\emptyset}/2)\\
	&= \sum_{i\in[n]} 1-D_b(1/2 + \hat{f}_{i,\emptyset}/2 \| 1/2) \label{eq:uni}\\
	& \leq \sum_{i\in[n]} 1-\frac{\hat{f}_{i,\emptyset}^2}{2\ln{2}} \label{eq:pinsk}\\
	& = n - \frac{1}{2\ln{2}}\sum_{i\in[n]} W^0[f_i] \label{eq:zerolev}.
\end{align}
where~\eqref{eq:uni} follows from $D(p \|u )=\log |\mathcal{X}|-H(X)$ where $u$ is the uniform pmf over $\mathcal{X}$,~\eqref{eq:pinsk} follows from Pinsker's inequality, and~\eqref{eq:zerolev} is from the fact that $W^0[f_i] = \hat{f}_{i,\emptyset}^2$ by definition. We therefore conclude that 
\begin{align}
\sum_{i\in[n]} W^0[f_i] \leq \epsilon\cdot 2\ln{2}.  
\end{align}		
The same derivations works for $g_i$, and the result follows. 
\end{proof}

Next, we show that most of the energy of each of these Boolean functions is concentrated on the first level, which will then imply their closeness to being some dictator function.  
\begin{lemma}\label{lem:almost_dictator}
	It holds that 
\begin{align}
	W^1[f_i] = 1-\varepsilon_i^f, \quad W^1[g_i] =1-\varepsilon_i^g, 
\end{align}
where $\varepsilon_i^f,\varepsilon_i^g\geq 0$ and 
\begin{align}
\sum_{i\in[n]} \left(\varepsilon_i^f + \varepsilon_i^g\right) \leq \frac{2\epsilon}{p(1-2p)}.
\end{align}
\end{lemma}
\begin{proof}
Write 
\begin{align}
\epsilon &\geq D(p_{X^nY^n} \| p_{XY}^\tn) \\
&\geq \E\left[d_H(f(U^n),g(V^n))\right]-np \label{eq:Lem5}\\
&=\sum_{i=1}^{n}\left(\Pr\left(f_i(U^n)\neq g_i(V^n)\right)-p\right)\\
&=\frac{1}{2}\sum_{i=1}^{n}\left(1-2p - \E[f_i(U^n)g_i(V^n)]\right)\\
&\geq\frac{1}{2}\sum_{i=1}^{n}\left(1-2p-\frac{1}{2}\sum_{k=0}^{n}(1-2p)^{k}\left(W^k[f_i]+W^k[g_i]\right)\right)\label{eq:Lem2}\\
&\geq -\epsilon \cdot \ln{2} + \frac{1}{2}\sum_{i=1}^{n}\left((1-2p)\left(1-\frac{1}{2}\left(W^1[f_i] + W^1[g_i]\right)\right) -(1-2p)^2\cdot \frac{1}{2}\sum_{k=2}^{n}\left(W^k[f_i]+W^k[g_i]\right)\right) \label{eq:Lem6}\\
& \geq -\epsilon \cdot \ln{2} + \frac{1-2p}{2}\sum_{i=1}^{n}\left(1-\frac{1}{2}\left(W^1[f_i] + W^1[g_i]\right) - (1-2p)\left(1-\frac{1}{2}\left(W^1[f_i] + W^1[g_i]\right)\right)\right) \label{eq:Foursum}\\
& = -\epsilon \cdot \ln{2} + p(1-2p)\sum_{i=1}^{n}\left(1-\frac{1}{2}\left(W^1[f_i] + W^1[g_i]\right)\right). 
\end{align}
where~\eqref{eq:Lem5} follows from Lemma~\ref{lem:Divbound},~\eqref{eq:Lem2} follows from Lemma~\ref{lem:corrbound},~\eqref{eq:Lem6} from Lemma~\ref{lem:almost_unbiased} and the fact that $\sum_{k=2}^{n}(1-2p)^kW^k[f]\leq (1-2p)^2\sum_{k=2}^{n}W^k[f]$ and~\eqref{eq:Foursum} follows from the fact that $\sum_{k=1}^{n}W^k[f]\leq 1$.  
Hence, we have that 
\begin{align}
\sum_{i=1}^{n}\left(W^1[f_i] + W^1[g_i]\right) &\geq 2n - \frac{1+\ln{2}}{p(1-2p)}\cdot \epsilon\\
& \geq 2n - \frac{2\epsilon}{p(1-2p)}.
\end{align}
The claim follows by recalling that since $f_i,g_i$ are Boolean functions, then $W^1[f_i],W^1[g_i] \leq 1$.
\end{proof}

Now, appealing to Lemma~\ref{lem:FKN}, we conclude that $f_i,g_i$ are close to some dictator function, i.e., for any $i\in[n]$ there exist $k_i, \ell_i \in [n]$ and $a_i,b_i\in \{-1,1\}$ such that 
\begin{align}
	\Pr(X_i\neq a_i\cdot U_{k_i}) &\leq K\cdot \varepsilon_i^f\\
	\Pr(Y_i\neq b_i\cdot V_{\ell_i}) &\leq K\cdot \varepsilon_i^g
\end{align}
for some universal constant $K$. In the next two lemmas, we show that these dictator functions are all distinct, and have the same signs for the same $i$. Namely, we show that the functions $f$ and $g$ are close to the same signed coordinate permutation of the Hamming cube. 
\begin{lemma} \label{lem:bijection}
	There exist a constant $K_p>0$ depending only on $p$, such that if $0\leq \epsilon < K_p$ then the mappings $k_i\leftarrow i$ and $\ell_i\leftarrow i$ are bijections from $[n]$ to $[n]$. 
\end{lemma}
\begin{proof}
	It suffices to rule our the case where $k_1=k_2=1$ and $a_1=a_2=1$. Assume toward contradiction that these equalities hold. Then 
	\begin{align}
		\Pr(X_1\neq X_2) &\leq \Pr(X_1\neq U_1 \vee X_2\neq U_1)\\
		& \leq \Pr(X_1\neq U_1) + \Pr(X_2\neq U_1)\\
		&\leq K\cdot (\varepsilon_1^f + \varepsilon_2^f)\\
		&\leq \frac{2K\epsilon}{p(1-2p)}.
	\end{align}
		This implies that for a sufficiently small $\epsilon$
	\begin{align}
		H(X^n) &\leq H(X_1|X_2) + \sum_{i=2}^n H(X_i)\\
		& \leq h_b(\Pr(X_1\neq X_2)) + n-1 \label{eq:Fano}\\
		& \leq  n- \left(1-h_b\left(\frac{2K}{p(1-2p)}\cdot \epsilon\right)\right), 
	\end{align}
	where we have used the convexity of the binary entropy function in~\eqref{eq:Fano}. But we know that $H(X^n) \geq n-\epsilon$, hence we arrive at a contradiction for any $\epsilon>0$ sufficiently small, as stated. 
\end{proof}

\begin{lemma} \label{lem:coordination}
	For any $0\leq \epsilon < K_p$, $k_i=\ell_i$ and $a_i=b_i$.
\end{lemma}
\begin{proof}
	It suffices to consider two cases. First, assume that $k_1=1, \ell_1=2$ and $a_1=b_1=1$. Define the events $E_1=\{X_1=U_1 \wedge Y_1=V_2\}$ and $E_2 = \{U_1\neq V_2\}$. Note that  
	\begin{align}
		\Pr(\overline{E_1})  &= \Pr(X_1\neq U_1 \vee Y_1\neq V_2)\\
		&\leq K(\varepsilon_1^f + \varepsilon_1^g)\\
		&\leq \frac{2K\epsilon}{p(1-2p)}.
	\end{align} 
	and $\Pr(\overline{E_2}) = 1/2$. We then have that 
	\begin{align}
		\Pr(X_1\neq Y_1) &\geq  \Pr(E_1\wedge E_2)\\
		&= 1- \Pr(\overline{E_1}\vee \overline{E_2})\\
		&\geq \frac{1}{2} - \frac{2K\epsilon}{p(1-2p)}.
	\end{align}
	We can then lower bound the divergence as follows, for $\epsilon$ small enough:
\begin{align}
	D(p_{X^nY^n} \| p_{XY}^\tn) &\geq D(p_{X_1Y_1}\|p_{XY})\\
	&  \geq \frac{2}{\ln{2}}\dtv^2(p_{X_1Y_1}, p_{XY}) \\&
	\geq \frac{2}{\ln{2}}\left(\Pr(X_1\neq Y_1) - \Pr(X\neq Y)\right)^2\\
	&\geq \frac{2}{\ln{2}}\left(\frac{1}{2} - \frac{2K}{p(1-2p)}\cdot \epsilon - p\right)^2.
\end{align}	
This is clearly larger than $\epsilon$ whenever the latter is sufficiently small, in contradiction.
The second case is where $k_1=\ell_1=1$ but $a_1=1, b_1=-1$. Define the events $E_1=\{X_1=U_1 \wedge Y_1=-V_1\}$ and $E_2 = \{U_1\neq -V_1\}$. Note that $\Pr(\overline{E_1}) \geq 1-p - \frac{2K\epsilon}{p(1-2p)}$ and $\Pr(\overline{E_2}) = p$. Thus $\Pr(X_1\neq Y_1) \geq 1-p - \frac{2K\epsilon}{p(1-2p)}$, and the proof follows similarly. 
\end{proof}

We have seen that in a $\epsilon$-successful deterministic simulation, the functions $f$ and $g$ are close to the same signed coordinate permutation of the Hamming cube, which is a stable variant of the case of exact simulation. The only thing left to show is the relation to the total variation distance. Assume without loss of generality that the coordinate permutation induced by $k_i\leftarrow i$ is the identity one, i.e., that $k_i=i$, and that $a_i=1$ for all $i$. Set the scalar noiseless channels
\begin{align}	
&q_{X_i|U_i}(x_i |v_i)=\ind (x_i=u_i)\\&q_{Y_i|V_i}(y_i | v_i)=\ind (y_i=v_i)
\end{align}
We can now upper bound the expected total variation distance by recalling Lemma~\ref{lem:TV_from_const}:
\begin{align}	
\E\dtv\left(p_{X^n|U^n}(\cdot\mid U^n)\,,\, \prod_{i=1}^nq_{X_i|U_i}(\cdot | U_i)\right)& =
\Pr (X^n\neq U^n)\\
&\leq \sum_{i=1}^{n}\Pr\left(X_i\neq U_i\right)\\
&\leq K\sum_{i=1}^{n}\varepsilon_i^f\\
&\leq \frac{2K\epsilon}{p(1-2p)}.
\end{align}
Finally, appealing to Markov's inequality, we conclude that the total variation distance converges in probability to zero.

\subsection{Randomized Schemes}
In this subsection, we treat the general case where Alice and Bob are allowed to use local randomness. 
\subsubsection{Almost Exact Simulation ($\epsilon>0, \delta=0$)}\label{subsec:almost_exact_sim_rand}
We begin with the case where we want to simulate the same $\dsbs(p)$, allowing a divergence of at most $\epsilon$, and also allowing local randomization. We proceed in steps: We first show (in Lemma~\ref{lem:probably_p}) that with high probability over the choice of local randomization, the expected Hamming distance between $X^n$ and $Y^n$ is close to $np$. We then show (in Lemma~\ref{lem:probably_dictator}) that this implies that the mappings $U^n\to X_i$ and $V^n\to Y_i$ are with high probability close to some dictator functions. Finally, we show that these mappings yield with high probability the same coordinate permutation (Lemma~\ref{lem:uniqueness}).

According to the functional representation lemma~\cite{el2011network}, for any pair of jointly distributed discrete r.v.s $(X,Y)$, one can write $Y$ as some deterministic function of $X$ and $Z$, where $Z$ is a discrete r.v. independent of $X$. In our case, this means that the simulating kernels $p_{X^n|U^n}$ and $p_{Y^n|V^n}$ can be replaced by random functions. i.e., we can write 
\begin{align}
X_i=f_i(U^n,A), \quad Y_i=g_i(V^n,B)
\end{align}
where $A,B$ and $(U^n,V^n)$ are mutually independent, and where $f_i(\cdot,a), g_i(\cdot,b)$ are Boolean functions. In the distributed simulation framework, $A$ and $B$  will denote the local randomnesses available to Alice and Bob, taking values in some general alphabets $\mathcal{A},\mathcal{B}$, respectively. 

\begin{lemma}\label{lem:probably_p}
	It holds that
	\begin{align}
	    \Pr\left[\frac{1}{n}\sum_{i=1}^{n}\Pr\left(X_i\neq Y_i \mid A,B\right) \leq p + \frac{O\left(\sqrt[4]{\epsilon}n^{-{3\over 4}}\right)}{1-2p} \right] = 1-O\left(\sqrt[4]{n\epsilon}\right)
	\end{align}
\end{lemma}
The outline of the proof is as follows: Showing that the average (over $i$) of $\Pr\left(X_i\neq Y_i \mid A,B\right)$ is with high probability not much larger than $p$ is equivalent to showing that the average (over $i$) of $\E(X_iY_i|AB)$ is with high probability not much smaller than $1-2p$.
To that end, consider the notion of \textit{conditional correlation}
\begin{align}
  \rho (X;Y|A,B) \triangleq \frac{\E(XY|A,B)-\E(X|A)\E(Y|B)}{\sqrt{\Var(X|A)\Var(Y|B)}},  
\end{align}
which is an r.v. that represents the correlation induced between $X$ and $Y$ when $A,B$ are randomly drawn.
We essentially show that we can replace the average (over $i$) of $\E(X_iY_i|A,B)$ (with high probability) with the average (over $i$) of the conditional correlations, and then use maximal correlation. In order for that to hold, we need first to show that $X_i,Y_i$ are almost unbiased, with high probability over the local randomness. First, we show that $X_i,Y_i$ are unconditionally almost unbiased. 
\begin{lemma} The following two claims hold: \label{lem:manouver}
\begin{enumerate}[(i)]
    \item $\frac{1}{2\ln 2}\sum_{i=1}^{n} \left(\E (X_iY_i)-(1-2\Pr(X\neq Y))\right)^2 \leq \epsilon $
    \item $\frac{1}{2\ln 2}\sum_{i=1}^{n} \E^2 (X_i)\leq \epsilon $
\end{enumerate}
\end{lemma}

Now, we show that $X_i,Y_i$ are almost unbiased conditioned on the local randomness, with high probability. 
\begin{lemma} \label{lem:power_bound}
Let $(X,Y,A,B)$ be jointly distributed r.v.s., where $X$ and $Y$ take values in $\{-1,1\}$, $A$ is independent of $B$, and $X-A-B$ and $Y-B-A$ form Markov chains. Suppose that 
    \begin{align}
    \max\left\{\left|\E(X)\right|,  \left|\E(Y)\right|, \left|\E (XY)-(1-2p) \right|\right\} \leq \epsilon \label{eq:pcorr}
    \end{align} 
    and
    \begin{align}\label{eq:bounded_rho_assumption} 
   \rho (X;Y|A,B)\leq 1-2p 
    \end{align}
    with probability $1$. Then it holds that
    \begin{align}
    \E \left[\E(X|A)^2+\E(Y|B)^2\right] \leq \frac{2\epsilon (1+\epsilon)}{1-2p}
    \end{align}
\end{lemma}
The proofs of these lemmas appear in Appendix~\ref{app:lem1} and~\ref{app:lem2}, respectively. We now continue to prove Lemma~\ref{lem:probably_p}.
\begin{proof}[Proof of Lemma~\ref{lem:probably_p}]
From Lemma~\ref{lem:manouver}, without loss of generality, we can write 
\begin{align} 
&\left|\E (X_iY_i)-(1-2p)\right|\leq \epsilon_i \label{eq:p_bound}\\&
\left|\E(X_i)\right| \leq \epsilon_i; \hspace{2mm} \left|\E(Y_i)\right| \leq \epsilon_i \label{eq:epsilon_bias}
\end{align} 
where $\sum_{i=1}^{n}\epsilon_i^2\leq \epsilon \cdot 2\ln{2}$, and hence by Jensen's inequality also $\sum_{i=1}^{n}\epsilon_i \leq \sqrt{n\epsilon\cdot 2\ln 2}$. In order to use Lemma~\ref{lem:power_bound} we still need to show the bound on the conditional correlation. This follows by noting that $\rho_m(X_i;Y_i|A,B)$, which is an r.v. that represents the maximal correlation corresponding to the distribution $P_{X_i,Y_i|A,B}(\cdot,\cdot\mid A,B)$, has for all $i \in [n]$
\begin{align}
 \rho_m(X_i;Y_i|A,B) & = \rho_m (f_i(U^n,A);g_i(V^n,B)|A,B) \\
 & \leq  \rho_m (U^n,A;V^n,B|A,B) \label{eq:DPI2} \\
 &=\max \{\rho_m (U^n;V^n),\rho_m (A;B)\} \label{eq:ten_mc} \\
 & = 1-2p  \label{eq:mc_BCS}
\end{align}
where~\eqref{eq:DPI2} follows from the DPI for maximal correlation,~\eqref{eq:ten_mc} follows from the tensorization property and~\eqref{eq:mc_BCS} follows since $(A,B)$ are independent and the maximal correlation of a $\dsbs(p)$ is $1-2p$. Hence, since $\rho(X_i;Y_i|A,B) \leq \rho_m(X_i;Y_i|A,B)$ by definition, we can use Lemma~\ref{lem:power_bound} to obtain
\begin{align}
\E \left(\sum_{i=1}^{n}\E(X_i|A)^2+\E(Y_i|B)^2\right) &\leq \sum_{i=1}^{n}\frac{2\epsilon_i(1+\epsilon_i)}{1-2p} \label{eq:power_bound}\\
&\leq \frac{\epsilon \cdot 4\ln{2}+2\sqrt{n\epsilon\cdot 2\ln 2}}{1-2p}\label{eq:sum_power_bound}
\end{align}
Finally, we show that we can replace the average (over $i$) of $\E(X_iY_i|A,B)$ (with high probability) with the average (over $i$) of the conditional correlations, and that the latter is about $1-2p$. To show that $\E(X_iY_i|A,B)$ can be replaced by the conditional correlation, let us denote $1-\eta_i^{AB}\triangleq \sqrt{\left(1-\E(X_i|A)^2\right)\left(1-\E(Y_i|B)^2\right)}$ and note that
\begin{align}
\sum_{i=1}^{n}\eta_i^{AB} &\leq n-\sum_{i=1}^{n}\min\left\{1-\E(X_i|A)^2,1-\E(Y_i|B)^2\right\}\\&=\sum_{i=1}^{n}\max\left\{\E(X_i|A)^2,\E(Y_i|B)^2\right\}\\& \leq \sum_{i=1}^{n}\E(X_i|A)^2+\E(Y_i|B)^2 \label{eq:bound_chiAB}
\end{align}
where we used $\sqrt{ab} \geq \min(a,b)$. Appealing to Markov's inequality, we have: 
\begin{align}
\Pr\left(\sum_{i=1}^{n}\E(X_i|A)^2+\E(Y_i|B)^2\geq \frac{\Omega\left(\sqrt[4]{n\epsilon}\right)}{1-2p} \right)= O\left(\sqrt[4]{n\epsilon}\right) \label{eq:sum_Markov}
\end{align}
This implies that with probability of at least $1-O\left(\sqrt[4]{n\epsilon}\right)$, the following bound holds:
 \begin{align}
\sum_{i=1}^{n} \rho(X_i;Y_i|A,B)=&\sum_{i=1}^{n} \rho(X_i;Y_i|A,B)\left(1-\eta_i^{AB}\right)+\rho(X_i;Y_i|A,B)\cdot \eta_i^{AB}
\\ = &\sum_{i=1}^{n}\E(X_iY_i|A,B)+\E(X_i|A)\E(Y_i|B)+\rho(X_i;Y_i|A,B)\cdot\eta_i^{AB} \\ \leq &
\sum_{i=1}^{n}\E(X_iY_i|A,B)+\E(X_i|A)\E(Y_i|B)+(1-2p)\sum_{i=1}^{n}\eta_i^{AB} \label{eq:Expected_bound}\\ \leq & \sum_{i=1}^{n}\E(X_iY_i|A,B)+\frac{O\left(\sqrt[4]{n\epsilon}\right)}{1-2p} \label{eq:eta_bound1}
\end{align}
\eqref{eq:Expected_bound} follows from using~\eqref{eq:mc_BCS}, and~\eqref{eq:eta_bound1} follows from~\eqref{eq:bound_chiAB} and~\eqref{eq:sum_Markov} together with the inequality $2\alpha\beta \leq \alpha^2+\beta^2$. 

We showed that the average (over $i$) of $\E(X_iY_i|A,B)$ can be replaced (with high probability) with the average (over $i$) of the conditional correlations, so now let us prove that the latter is roughly $1-2p$. In a similar fashion to the above, consider 
\begin{align}
\E\left[\sum_{i=1}^{n}\rho(X_i;Y_i|A,B)\right] &= \sum_{i=1}^{n}\E\left(\E(X_iY_i|A,B)+\E(X_i|A)\E(Y_i|B)\right) +\sum_{i=1}^{n}\E\left(\rho(X_i;Y_i|A,B)\cdot \eta_i^{AB}\right)\\ &=\sum_{i=1}^{n}\E(X_iY_i)+\E(X_i)\E(Y_i) +\sum_{i=1}^{n}\E\left(\rho(X_i;Y_i|A,B)\cdot \eta_i^{AB}\right)\\ &\geq 
n(1-2p)-O\left(\sqrt{n\epsilon}\right) -(1-2p)\E\left[\sum_{i=1}^{n} \eta_i^{AB}\right] \label{eq:e_bound}
\\ &\geq  n(1-2p)-O\left(\sqrt{n\epsilon}\right)
\label{eq:use_eta_bound}
\end{align}
where~\eqref{eq:e_bound} follows from~\eqref{eq:epsilon_bias} and \eqref{eq:mc_BCS}, and~\eqref{eq:use_eta_bound} follows from~\eqref{eq:sum_power_bound} and \eqref{eq:bound_chiAB}. Now, note that the r.v. $Z\triangleq n(1-2p)-\sum_{i=1}^{n}\rho(X_i;Y_i|A,B)$ is non-negative (due to~\eqref{eq:mc_BCS}) and its expectation is $\E(Z)\leq O\left(\sqrt{n\epsilon}\right)$. Hence we conclude from Markov's inequality that 
\begin{align}
\Pr \left(n(1-2p)-\sum_{i=1}^{n}\rho(X_i;Y_i|A,B)\geq \Omega \left(\sqrt[4]{n\epsilon}\right)\right)= O\left(\sqrt[4]{n\epsilon }\right)\label{eq:cond_corr_bound}
\end{align}
From the union bound applied to the events in~\eqref{eq:cond_corr_bound} and~\eqref{eq:eta_bound1}, we see that
\begin{align}
\Pr\left(\sum_{i=1}^{n}\E(X_iY_i|AB) \geq  n(1-2p)- O\left(\sqrt[4]{\epsilon n}\right)\right) = 1-O\left(\sqrt[4]{\epsilon n}\right), 
\end{align}
The lemma now follows by substituting $\E(X_iY_i|AB)=1-2\Pr(X_i\neq Y_i|A,B)$.
\end{proof}
Let us know define $f_i^a \triangleq f_i(\cdot, a)$ and $g_i^b \triangleq g_i(\cdot, b)$ as the Boolean functions used to generate $X_i$ and $Y_i$ respectively, when $(A,B)=(a,b)$.
In the next lemma, we show that w.h.p. most of the energy of these Boolean functions is concentrated on the first level, which will then imply their closeness to being some dictator function.  

\begin{lemma} \label{lem:probably_dictator}
\begin{align}
    \Pr \left(\frac{1}{n} \sum_{i=1}^{n}\frac{W^1[f_i^A] + W^1[g_i^B]}{2} \geq 1 - \frac{O\left(\sqrt[4]{\epsilon}n^{-{3 \over 4}}\right)}{p(1-2p)^2}\right) = 1-O\left(\sqrt[4]{n\epsilon}\right).
\end{align}
\end{lemma}

\begin{proof}
Lemma~\ref{lem:probably_p} implies that there exists a subset $S\subseteq \mathcal{A}\times \mathcal{B}$ with $\Pr((A,B)\in S) = 1-O\left(\sqrt[4]{n\epsilon }\right)$, such that for any $(a,b)\in S$ the following two inequalities hold: 
\begin{align}
&\sum_{i=1}^{n}W^0[f_i^a]+W^0[g_i^b]\leq \frac{O\left(\sqrt[4]{n\epsilon}\right)}{1-2p}\\&   \sum_{i=1}^{n}\left(1-2p-\E(f_i^a(U^n)g_i^b(V^n))\right) \leq  \frac{O\left(\sqrt[4]{\epsilon n}\right)}{1-2p}\label{eq:error_prob_bound}
\end{align}
Note that this is essentially the same inequality as in  Lemma~\ref{lem:almost_dictator}, eq.~\eqref{eq:Lem2} with $\epsilon$ replaced by the right-hand side of~\eqref{eq:error_prob_bound}, so we can continue along the same line of proof, which results in:   
\begin{align}
\sum_{i=1}^{n}\left(W^1[f_i^a] + W^1[g_i^b]\right) \geq  2n-\frac{O \left(\sqrt[4]{n\epsilon}\right)}{p(1-2p)^2}.
\end{align}
and the claim follows.
\end{proof}
By combining lemma~\ref{lem:probably_dictator} and lemma~\ref{lem:FKN}, we conclude that, for all $(a,b)\in S$, $f_i^a,g_i^b$ are close to some dictator function, i.e., for any $i\in[n]$ there exist $k_i^a, \ell_i^b \in [n]$ and $\alpha_i^a,\beta_i^b\in \{-1,1\}$ such that
\begin{align}
\Pr(f_i^a(U^n)\neq \alpha_i^a\cdot U_{k_i^a}) &\leq  \varepsilon_i^{f_a} \label{eq:dictatorA}\\
\Pr(g_i^b(V^n)\neq \beta_i^b\cdot V_{\ell_i^b}) &\leq  \varepsilon_i^{f_a}\label{eq:dictatorB}
\end{align}
for some set of constants $\varepsilon_i^{f_a},\varepsilon_i^{g_b}\geq 0$ that has $\sum_{i=1}^{n}\left(\varepsilon_i^{f_a} + \varepsilon_i^{g_b}\right)\leq \frac{O\left(\sqrt[4]{n\epsilon}\right)}{p(1-2p)^2}$. 
In the next lemmas, we show that with high probability the functions $f_i^A$ and $g_i^B$ are close to the same signed coordinate permutation of the Hamming cube, that this coordinate permutation does not depend on the local randomness and that it is unique. 

\begin{lemma} \label{lem:uniqueness}
	For all $(a,b)\in S$, it holds that
\begin{align}
k_i^a = \ell_i^b, \alpha_i^a=\beta_i^b.
\end{align}
Moreover, $k_i^a \neq k_j^a$ for $i\neq j$. 
\end{lemma}
\begin{proof}
From  lemma~\ref{lem:coordination}, if $(a,b) \in S$ and $k_i^a \neq \ell_i^b$, then  
\begin{align}
    \Pr(X_i\neq Y_i) \geq \frac{1}{2} - \left(\varepsilon_i^{f_a}+\varepsilon_i^{g_b}\right)
\end{align}
In a similar fashion, if $k_i^a = \ell_i^b$, $\alpha_i^a \neq \beta_i^b$, then 
\begin{align}
    \Pr(X_i\neq Y_i) \geq 1-p - \left(\varepsilon_i^{f_a}+\varepsilon_i^{g_b}\right)
\end{align} 
It follows that $\Pr(X_i\neq Y_i)-p\geq \frac{1}{2}-p-\frac{O \left(\sqrt[4]{n\epsilon}\right)}{p(1-2p)^2}$, in contradiction to lemma~\ref{lem:probably_p} for $n\epsilon$ small enough, implying $k_i^a = \ell_i^b$, $\alpha_i^a=\beta_i^b$.  

In order to prove uniqueness, we first show that the signed permutation does not depend on the local randomization. Let us define the pair of r.v.s $(k_i^A,\ell_i^B)$ in the following manner: When $(a,b)\notin S$, these r.v.s are some arbitrary index in $[n]$, and when $(a,b)\in S$ they are both equal to the same index $k_i^a$. Hence, from lemma~\ref{lem:probably_dictator} and the distribution of $(k_i^A,\ell_i^B)$ we have:
\begin{align}
    1-O\left(\sqrt[4]{n\epsilon}\right)&= \Pr\left(k_i^A=\ell_i^B\right)\\& = \sum_{j=1}^{n}\Pr\left(k_i^A=j\right)\Pr\left(\ell_i^B=j\right) \label{eq:a_b_independent}\\ &\leq \underset{j}{\max}\Pr\left(k_i^A=j\right) \sum_{j=1}^{n}\Pr\left(\ell_i^B=j\right)\\&=\underset{j}{\max}\Pr\left(k_i^A=j\right),\label{eq:max_index}
\end{align}	
Where~\eqref{eq:a_b_independent} follows from the independence of $A$ and $B$. This means $k_i^A$ (resp. $\ell_i^B$) will be constant with high probability, and the same holds for $\alpha_i^A$ (resp. $\beta_i^B$).
 We now show that from this simple fact it follows that, with high probability, $k_i\leftarrow i$ is a bijection from $[n]$ to $[n]$. Assume without loss of generality that the index that maximizes~\eqref{eq:max_index} for both $k_1$ and $k_2$ is $1$, and that also $\alpha_1=\alpha_2=1$. Then, according to Lemma~\ref{lem:bijection},
	\begin{align}
	\Pr(X_1\neq X_2|(A,B)\in S) \leq \frac{O \left(\sqrt[4]{n\epsilon}\right)}{p(1-2p)^2},
	\end{align}
	and since $\Pr((A,B)\in S) = 1-O\left(\sqrt[4]{n\epsilon }\right)$, it follows from total probability that $\Pr(X_1\neq X_2) \leq \frac{O \left(\sqrt[4]{n\epsilon}\right)}{p(1-2p)^2}$.
	Using the divergence bound of Lemma~\ref{lem:Divbound2}, we have
	\begin{align}
 	\epsilon &\geq D\left( p_{X^nY^n}\| p^\tn_{XY}\right)\\&\geq I(X_1;X_2)\\&=H(X_1)-H(X_1|X_2)\\&\geq h_b \left(\frac{1}{2}+\sqrt{\epsilon \cdot 2\ln2}\right)-\Pr(X_1\neq X_2) -h_b (\Pr(X_1\neq X_2)) \label{eq:fano2}\\&\geq 1-4\epsilon-\frac{O \left(\sqrt[4]{n\epsilon}\right)}{p(1-2p)^2}-h_b \left(\frac{O \left(\sqrt[4]{n\epsilon}\right)}{p(1-2p)^2}\right)\label{eq:binary_ineq}
	\end{align}
	where we have used Fano's inequality and the fact that lemma~\ref{lem:manouver} implies $\left|\Pr\left(X_i=\frac{1}{2}\right)-\frac{1}{2}\right|\leq \sqrt{\epsilon \cdot 2\ln2}$ in~\eqref{eq:fano2} and the inequality $h_b\left(\frac{1+x}{2}\right)\geq 1-\frac{x^2}{2}\log e,|x|<1$ in~\eqref{eq:binary_ineq}. It is evident that when $\epsilon$ is small enough we arrive at a contradiction. 
\end{proof}
In order to show the relation to the total variation distance, again we set the scalar noiseless channels
\begin{align}	
&q_{X_i|U_i}(x_i |v_i)=\ind (x_i=u_i)\\&q_{Y_i|V_i}(y_i | v_i)=\ind (y_i=v_i)
\end{align}
Appealing again to lemma~\ref{lem:TV_from_const}, we can write the total variation distance as
\begin{align}	
\E\dtv\left(p_{X^n|U^n}(\cdot\mid U^n)\,,\, \prod_{i=1}^nq_{X_i|U_i}(\cdot | U_i)\right)& =\Pr (X^n\neq U^n)
\\&= \frac{O\left(\sqrt[4]{n\epsilon}\right)}{p(1-2p)^2},
\end{align}
hence when $\epsilon =  o(1/n)$ the total variation distance goes to zero in probability.

\subsubsection{The General Symmetric Case ($\epsilon>0, \delta>0$)}\label{subsec:theorem2}
	Lemma~\ref{lem:manouver} now implies that 
	\begin{align}
    |\E (X_iY_i)-(1-2p)| \leq \epsilon_i+2\delta
	\end{align}
	which means we can incorporate $\delta$ into the simulation distortion $\epsilon$, and since $\sum_{i=1}^{n} \epsilon_i^2 \leq  \epsilon =  o(1/n)$ due to the condition on $\epsilon$, we want also that $\sum_{i=1}^{n} \delta^2 = o(1/n)$ so that the conditions for successful simulation will hold, making $\delta=o(1/n)$ and also
	\begin{align} 
	\E \dtv\left(p_{\sigma(X^n)|U^n}(\cdot\mid U^n)\,,\, \prod_{i=1}^nq_{X_i|U_i}(\cdot | U_i)\right)= \frac{O\left(\sqrt{n\delta+\sqrt{n\epsilon}}\right)}{p(1-2p)^2}.
	\end{align}
Appealing to the Markov inequality once again, we conclude the direct part of Theorem~\ref{thrm:DSBS}. A counterexample for when $\epsilon=\omega (1/\sqrt{n})$ or $\delta=\omega (1/\sqrt{n})$ is provided in the following section. 
	
\subsection{Counterexample: $\epsilon=\omega (1/\sqrt{n}), \delta=\omega (1/\sqrt{n})$, with total variation bounded away from zero}\label{subsec:converse}
Consider some real numbers $0<\alpha<\beta<1$ such that $\alpha+\beta=1$, and suppose we are interested in simulation of a $\dsbs(p+n^{-\alpha})$. We propose the following algorithm: Partition the $n$ length sequence into $n^{\alpha}$ disjoint subsequences of length $n^{\beta}$, marked $S_1$ to $S_{n^{\alpha}}$. 
Set $X_i=U_i$ to be a clean channel for all $i$. $Y_i$ is determined in the following manner: If $i \in S_j$ and is also not the last coordinate in the subsequence, then $Y_i$ is a scalar channel $Y_i=V_i\tilde{Z}_i$. If $i $ is the last coordinate in the subsequence $S_j$, then
\[Y_i= \left\{ \begin{array}{ll}
V_i, & if \hspace{2mm} r_j=0,\\
\prod_{\underset{k \neq i}{k \in S_j}}V_k, & if \hspace{2mm} r_j=1 ,\end{array} \right. \]
where:
\begin{enumerate}
    \item $\{r_j\}_{j=1}^{n^{\alpha}}$ is an i.i.d. sequence of $\bern{\left(\mu\right)}$ r.v.s.
    \item $\{\tilde{Z}_i\}_{i=1}^n$ is an i.i.d. sequence of $\bern{\left(q\right)}$ r.v.s.
    \item $\{r_j\}_{j=1}^{n^{\alpha}}, \{\tilde{Z}_i\}_{i=1}^n$ are independent of each other and of $\{V_i\}_{i=1}^n$.
    \item $q,\mu$ are real parameters.
\end{enumerate}
Let us pick $q,\mu$ such that $\Pr(X_i\neq Y_i) = p+n^{-\alpha}$, i.e.,
\begin{align}
p \ast q= p+n^{-\alpha} &\Longrightarrow q=\frac{n^{-\alpha}}{1-2p}, \\ \left(1-\mu\right)p+\mu \cdot \frac{1}{2}=p+n^{-\alpha} &\Longrightarrow \mu=\frac{2n^{-\alpha}}{1-2p}.
\end{align}
We analyze this example using Lemma~\ref{lem:Divbound2}, which showed that
\begin{align}
D\left( p_{X^nY^n}\| p^\tn_{XY}\right) =  \sum_{i=1}^{n}I\left(X_i,Y_i;X^{i-1},Y^{i-1}\right) + \sum_{i=1}^{n}D\left( p_{X_iY_i}\| p_{XY}\right). \label{eq:div_muI} 
\end{align}
With the above choice of $q,\mu$, the divergence expression in~\eqref{eq:div_muI} is zeroed. Furthermore, all pairs $(X_i,Y_i), i \in S_k$,$(X_j,Y_j), j \in S_m, m\neq k$ are independent, so the mutual information in~\eqref{eq:div_muI} is reduced to:
\begin{align}
\sum_{k=1}^{n^{\alpha}}\sum_{i\in S_k }I\left(X_i,Y_i;X^{i-1},Y^{i-1}\right)
=&\sum_{k=1}^{n^{\alpha}}\sum_{i\in S_k }I\left(U_i,Y_i;U^{i-1},Y^{i-1}\right)\\=&\sum_{k=1}^{n^{\alpha}}I\left(Y_{kn^{\beta}};U^{kn^{\beta}-1}\left|U_{kn^{\beta}}\right.\right)+I\left(U_{kn^{\beta}};Y^{kn^{\beta}-1}\left|U^{kn^{\beta}-1}\right.\right)\\&+I\left(Y_{kn^{\beta}};Y^{kn^{\beta}-1}\left|U^{kn^{\beta}}\right.\right),
\end{align}
where the last equality follows since all coordinates inside a subsection, excluding the last coordinate, form an i.i.d. set. Since the distribution within any subsection is the same, we need only show that
\begin{align}
I\left(Y_{n^{\beta}};U^{n^{\beta}-1}\left|U_{n^{\beta}}\right.\right)+I\left(U_{n^{\beta}};Y^{n^{\beta}-1}\left|U^{n^{\beta}-1}\right.\right)+I\left(Y_{n^{\beta}};Y^{n^{\beta}-1}\left|U^{n^{\beta}}\right.\right)
\end{align}  
goes to zero fast enough. It is clear that $U_{n^{\beta}}$ is independent of $\left(U^{n^{\beta}-1},Y^{n^{\beta}-1}\right)$, hence the second term is zero. For the first term, we claim that $U^{n^{\beta}-1}$ is independent of $\left(U_{n^{\beta}},Y_{n^{\beta}}\right)$ in the limit of large $n$. This follows since $Y_{n^\beta}$ is either $V_{n^\beta}$, or 
\begin{align}
\prod_{\underset{k \neq n^\beta}{k \in S_1}}V_k=\prod_{k=1}^{n^{\beta}-1}U_k\prod_{k=1}^{n^{\beta}-1}Z_k=\prod_{k=1}^{n^{\beta}-1}U_k\hat{Z}_1,\label{eq:independent_Unbeta}
\end{align}
where $\hat{Z}_1 \sim \bern \left(\frac{1}{2}\left(1-(1-2p)^{n^\beta-1}\right)\right)$, implying that $\hat{Z}_1$ approaches $\bern\left(\frac{1}{2}\right)$ in distribution exponentially fast in $n^\beta$. For the third term, we claim that $\left(U_{n^\beta},Y_{n^\beta}\right)$ are independent of $\left(U^{n^\beta-1},Y^{n^\beta-1}\right)$ in the limit of large $n$. This follows since we can write~\eqref{eq:independent_Unbeta} as
\begin{align}
\prod_{\underset{k \neq n^\beta}{k \in S_1}}V_k=\prod_{k=1}^{n^{\beta}-1}Y_k\prod_{k=1}^{n^{\beta}-1}\tilde{Z}_k=\prod_{k=1}^{n^{\beta}-1}Y_k\hat{Z}_2,
\end{align}
and since $\tilde{Z}_i \sim \bern\left(\frac{n^{-\alpha}}{1-2p}\right)$, we have that $\hat{Z}_2$ approaches $\bern\left(\frac{1}{2}\right)$ in distribution exponentially fast in $n^{\beta-\alpha}$. However, it is clear that with some positive probability, at least one of the coordinates $Y_i$ is the parity of $n^\beta$ bits of $V^n$, hence the total variation of $Y^n$ from any memoryless channel is bounded away from zero. 

This scheme can also be used for the case when $\epsilon$ and $\delta$ are swapped, i.e., when $\delta=0,\epsilon=n^{-\zeta}$ for some $\zeta<0.5$. Let us pick this time $0.5<\alpha<\beta<1$ and set $q=n^{-\alpha}$. This time we partition the $n$ length sequence into $n^{1-\beta}$ disjoint subsequences of length $n^{\beta}$ and use the same law as before, with the adjustment of $\mu=n^{-1+\beta}$. Then the conditions for the mutual information to zero out and for the total variation of $Y^n$ from any memoryless channel to be bounded away from zero still hold. However, the KL divergence of~\eqref{eq:div_muI} is not zero, since there is a discrepancy between $\Pr(X_i\neq Y_i)$ and $p$. It is well known that, for small enough $q$, 
\begin{align}
D\left( p+q\| p\right) \approx \frac{q^2}{p(1-p)}.
\end{align}
As a consequence, the divergence is about $n^{-2\alpha}$ for $n-n^{1-\beta}$ coordinates, and for the other $n^{1-\beta}$ the divergence is about $n^{-2+2\beta}$, making the total sum result in 
\begin{align}
\left(n-n^{-1+\beta}\right)\cdot n^{-2\alpha}+n^{1-\beta}\cdot n^{-2+2\beta}=n^{-1+\beta}\left(1-n^{-2\alpha}+n^{1-\beta-2\alpha}\right) \approx n^{-1+\beta},
\end{align}
implying that any $\zeta=1-\beta<0.5$ achieves the converse, thus concluding our example.

\section{Summary and Discussion}\label{sect:Sum}
We considered a distributed source simulation problem: Alice and Bob, observing two jointly distributed i.i.d. sequences according to some $p_{UV}$, are required to simulate two jointly distributed i.i.d. sequences according to some $p_{XY}$, with no communication between them and no shared randomness. Motivated by Wyner's result for centralized source simulation, we were able to characterize a new region of simulable distributions $\mathcal{S}(p_{UV})$, which integrates Wyner's digital scheme with an analog scheme. This hybrid construction allowed us to achieve a generally larger set of simulable distributions than the union of digital and analog schemes, but due to the hybrid nature of our scheme, our simulable region is nontrivial only in the case where $\cgk(U,V)$ is positive. In other words, when $U$ and $V$ lack a common part, the agents cannot cooperate via codebooks, leaving them with only the analog option. This brought us to conjecture that if $\cgk(U,V)=0$, then truly only analog simulation is possible. This conjecture proves very difficult to verify, mainly due to the difficulty in formulating a measure of closeness between a general function and a scalar function and determining unequivocally whether a distribution achieved via vector simulation is outside the analog simulation achievable region. Hence, we addressed the $\dsbs$ case, specifically the simulation of a $\dsbs(p+\delta)$ from a $\dsbs(p)$. For this case, it is known that $\delta \in[0,1-2p]$ is both necessary and sufficient, and can be attained by a scalar Markov chain. We showed that if $\delta$ and the simulation distortion $\epsilon$ are taken to be small enough, then any successful simulation will be close to scalar, in the sense that it would be virtually impossible to tell it apart from a scalar one with any statistical test. While that result is well known for the case of $\epsilon=\delta=0$, we extended it to the case of $\epsilon,\delta=o(1/n)$, and further showed that this is close to being tight.
\section{Appendix}\label{sect:Appendix}
\subsection{Proof of Lemma~\ref{lem:manouver}} \label{app:lem1}
\begin{proof}
\begin{align}
 \epsilon &\geq \sum_{i=1}^{n}D\left( p_{X_iY_i}\| p_{XY}\right) \label{eq:SumDiv}\\
&\geq \sum_{i=1}^{n}D\left( \Pr(X_i\neq Y_i)\parallel \Pr(X \neq Y)\right) \label{eq:DPI}\\
&\geq \frac{2}{\ln 2} \sum_{i=1}^{n}\dtv^2(\Pr(X_i\neq Y_i), \Pr(X\neq Y)) \label{eq:pinsk2}\\
&= \frac{2}{\ln 2} \sum_{i=1}^{n} \left(\Pr(X_i \neq Y_i)-p\right)^2 \\&=\frac{1}{2\ln 2}\sum_{i=1}^{n} \left(\E (X_iY_i)-(1-2p)\right)^2  
\end{align}
where~\eqref{eq:SumDiv} follows from Lemma~\ref{lem:Divbound2} ,~\eqref{eq:DPI} follows from the data-processing inequality for divergences and~\eqref{eq:pinsk2} is from Pinsker's inequality. In a similar fashion:
\begin{align}
\epsilon &\geq \sum_{i=1}^{n}D\left( p_{X_i}\parallel \frac{1}{2}\right)\\ 
&\geq \frac{2}{\ln 2}\sum_{i=1}^{n} \left(\Pr(X_i=1)-\frac{1}{2}\right)^2 \\&=\frac{1}{2\ln 2}\sum_{i=1}^{n} \E^2 (X_i)
\end{align}
\end{proof}

\subsection{Proof of Lemma~\ref{lem:power_bound}}\label{app:lem2}
\begin{proof} Using the Markov chains $X-A-B$ and $Y-B-A$, and recalling that $X,Y\in\{-1,1\}$, we can rewrite~\eqref{eq:bounded_rho_assumption} as 
\begin{align}
\E(XY|AB)\leq \sqrt{\left(1-\E(X|A)^2\right)\left(1-\E(Y|B)^2\right)}(1-2p)+\E(X|A)\E(Y|B)
\end{align}
with probability $1$. By taking the expectation on both sides, we get 
\begin{align}
\E(XY)& \leq \E \left[\sqrt{\left(1-\E(X|A)^2\right)\left(1-\E(Y|B)^2\right)}\right](1-2p)+\E \left[\E(X|A)\right]\E \left[\E(Y|B)\right]\label{eq:A_and_B_are_indp6}
\\ &\leq \E \left[\sqrt{\left(1-\E(X|A)^2\right)\left(1-\E(Y|B)^2\right)}\right](1-2p)+|\E(X)|\cdot |\E(Y)|\\
&\leq \E \left[\frac{1-\E(X|A)^2+1-\E(Y|B)^2}{2}\right](1-2p)+\epsilon^2 \label{eq:arigeo1} \\
&=1-2p-\frac{1-2p}{2}\E \left[\E(X|A)^2+\E(Y|B)^2\right]+\epsilon^2. 
\end{align}
We have used the assumption that $A,B$ are independent in~\eqref{eq:A_and_B_are_indp6}, and~\eqref{eq:pcorr} together with the arithmetic-geometric mean inequality in~\eqref{eq:arigeo1}. Rearranging the above and using~\eqref{eq:pcorr} again, the result follows.
\end{proof}

\subsection{Alternative proof of Lemma~\ref{lem:soft}}\label{app:soft}
\begin{proof}
	We draw a random encoding $A$ independently of $U^n$, where $A(u^n)\sim p_{W|U}^\tn(\cdot|u^n)$ and all the encoded vectors are mutually independent. We will show that 
	\begin{align}
		\lim_{n\to\infty}\E_A D\left(p_{X^n|A}\parallel p_X^\tn\right) = 0,
	\end{align}
	This would immediately imply the existence of a desired sequence of encodings. 
	
	Set some $\tau>0$ and define the following typical set:
	\begin{align}
	\mathcal{A}_{\tau}^{(n)} \dfn \left\{(x^n,u^n,w^n):\frac{1}{n}\log \left(\frac{p_U^\tn(u^n)\cdot p_{X|UW}^\tn(x^n|u^n,w^n)}{p_X^\tn(x^n)}\right) < \tau \right\} 
	\end{align}
	For later reference, note that we can upper bound the indicator function of the typical set by 
	\begin{align}\label{eq:typ_indbound}
	\ind_{\mathcal{A}_{\tau}^{(n)}}(x^n,u^n,w^n) \leq \left(\frac{2^{n\tau}}{\frac{p_U^\tn(u^n)\cdot p_{X|UW}^\tn(x^n|u^n,w^n)}{p_X^\tn(x^n)}}\right)^\beta.
	\end{align}
	for any  $\beta>0$. 
	
	We now separate the contribution to $p(x^n|a$) coming from typical triplets $(x^n,u^n,a(u^n)) \in\mathcal{A}_{\tau}^{(n)}$ and atypical triplets. To that end, define the following functions:
	\begin{align}
	p_1(x^n|a)=\sum_{u^n}p_U^\tn(u^n)p_{X|UW}^\tn(x^n|u^n,a(u^n))\ind_{\mathcal{A}_{\tau}^{(n)}}(x^n,u^n,a(u^n)),
	\end{align}
	\begin{align}
	p_2(x^n|a)=\sum_{u^n}p_U^\tn(u^n)p_{X|UW}^\tn(x^n|u^n,a(u^n))(1-\ind_{\mathcal{A}_{\tau}^{(n)}}(x^n,u^n,a(u^n))).
	\end{align}
	where $\ind_{\mathcal{A}_{\tau}^{(n)}}$ is the indicator function for the set $\mathcal{A}_{\tau}^{(n)}$. Note that $p(x^n|a) = p_1(x^n|a) + p_2(x^n|a)$, and that by construction 
	\begin{align}
	\E(p_1(x^n|A)+p_2(x^n|A))&=\E_A p(x^n|A) = p_X^\tn(x^n),\label{eq:ratio}
	\end{align}
	where the expectation is taken over the random choice of $A$. The divergence between the distribution of $X^n$ and the desired distribution given the encoding is:
	\begin{align}
	D(p_{X^n|A}(\cdot\mid a)\parallel p_X^\tn)&=\sum_{x^n}p(x^n|a)\log\frac{p(x^n|a)}{p_X^\tn(x^n)}\\
	&= \sum_{x^n}(p_1(x^n|a)+p_2(x^n|a))\log\frac{p_1(x^n|a) + p_2(x^n|a)}{\alpha p_X^\tn(x^n) + (1-\alpha)p_X^\tn(x^n)}\\
	\label{eq:eq1}&\leq\sum_{x^n}p_1(x^n|a)\log\frac{p_1(x^n|a)}{\alpha p_X^\tn(x^n)}+p_2(x^n|a)\log\frac{p_2(x^n|a)}{(1-\alpha)p_X^\tn(x^n)}\\
	& = \sum_{x^n}p_1(x^n|a)\log\frac{p_1(x^n|a)}{p_X^\tn(x^n)}+\sum_{x^n}p_2(x^n|a)\log\frac{p_2(x^n|a)}{p_X^\tn(x^n)} \\ &\quad +\log\left(\frac{1}{\alpha}\right)\sum_{x^n}p_1(x^n|a) + \log\left(\frac{1}{1-\alpha}\right)\left(1-\sum_{x^n}p_1(x^n|a)\right).
	\end{align}	
	for any $\alpha\in[0,1]$, where in~\eqref{eq:eq1} we have used the log-sum inequality. We can now minimize the bound by choosing $\alpha=\sum_{x^n}p_1(x^n|a)$, which yields 
	\begin{align}\label{eq:div_bound}
	D(p_{X^n|A}(\cdot\mid a)\parallel p_X^\tn) \leq \underbrace{\sum_{x^n}p_1(x^n|a)\log\frac{p_1(x^n|a)}{p_X^\tn(x^n)}}_{\dfn g_1(a)}
	+
	\underbrace{\sum_{x^n}p_2(x^n|a)\log\frac{p_2(x^n|a)}{p_X^\tn(x^n)}}_{\dfn g_2(a)} 
	+ 
	\underbrace{h\left(\sum_{x^n}p_2(x^n|a)\right)}_{\dfn g_3(a)}.
	\end{align}	
	where $h(\delta) = -\delta\log\delta - (1-\delta)\log(1-\delta)$ is the binary entropy function.  
	
We proceed by taking the expectation of both sides of~\eqref{eq:div_bound} over the choice of the encoding $A$, and examine the expectations of the three components $\E g_1(A),\E g_2(A)$ and $\E g_3(A)$, showing they all approach zero. For the first component, we have: 
	\begin{align}
	\E g_1(A) &= \E\left[\sum_{x^n}p_X^\tn(x^n)\frac{p_1(x^n|A)}{p_X^\tn(x^n)}\log\frac{p_1(x^n|A)}{p_X^\tn(x^n)}\right]\\
	& = \sum_{x^n}p_X^\tn(x^n)\E\left[\frac{p_1(x^n|A)}{p_X^\tn(x^n)}\log\frac{p_1(x^n|A)}{p_X^\tn(x^n)}\right]\\
	\label{eq:log_bound}&\leq \log (e) \sum_{x^n}p_X^\tn(x^n)\E \left[\frac{p_1(x^n|A)}{p_X^\tn(x^n)}\left(\frac{p_1(x^n|A)}{p_X^\tn(x^n)}-1\right)\right]\\
	&=\log (e)\sum_{x^n}p_X^\tn(x^n)\left[ \E \left[\left(\frac{p_1(x^n|A)}{p_X^\tn(x^n)}\right)^2\right]-\E \left[\frac{p_1(x^n|A)}{p_X^\tn(x^n)}\right]\right]\\ 
	\label{eq:avg_less_than_1}&\leq \log (e)\sum_{x^n}p_X^\tn(x^n)\left[\E \left[\left(\frac{p_1(x^n|A)}{p_X^\tn(x^n)}\right)^2\right]-\left(\E \left[\frac{p_1(x^n|A)}{p_X^\tn(x^n)}\right]\right)^2\right]\\
	&= \log (e)\sum_{x^n}p_X^\tn(x^n)\Var\left(\frac{p_1(x^n|A)}{p_X^\tn(x^n)}\right)\\
	\label{eq:g1_bound}&=\log(e)\sum_{x^n}\frac{\Var(p_1(x^n|A))}{p_X^\tn(x^n)}
	\end{align}
	where~\eqref{eq:log_bound} follows from the inequality $\log(x)\leq (x-1)\log (e)$, and~\eqref{eq:avg_less_than_1} follows by noting that~\eqref{eq:ratio} implies $\E\left(\frac{p_1(x^n|A)}{p_X^\tn(x^n)}\right)\leq 1$.	
	Let us now upper bound $\Var(p_1(x^n|A))$. 
	\begin{align} 
	\label{eq:additive_var}
	\Var(p_1(x^n|A)) &= \Var\left(\sum_{u^n}p_U^\tn(u^n)p_{X|UW}^\tn(x^n|u^n,A(u^n))\ind_{\mathcal{A}_{\tau}^{(n)}}(x^n,u^n,A(u^n))\right) 
	\\&=\sum_{u^n}\Var\left(p_U^\tn(u^n)p_{X|UW}^\tn(x^n|u^n,A(u^n))\ind_{\mathcal{A}_{\tau}^{(n)}}(x^n,u^n,A(u^n))\right)\\
	&\leq\sum_{u^n}\E \left[p_U^\tn(u^n)p_{X|UW}^\tn(x^n|u^n,A(u^n))\ind_{\mathcal{A}_{\tau}^{(n)}}(x^n,u^n,A(u^n))\right]^2 \\
	&=\sum_{u^n,w^n}p_{W|U}^\tn(w^n|u^n)\left[p_U^\tn(u^n)p_{X|UW}^\tn(x^n|u^n,w^n)\right]^2\ind_{\mathcal{A}_{\tau}^{(n)}}(x^n,u^n,w^n)\\
	&\leq\sum_{u^n,w^n}p_{W|U}^\tn(w^n|u^n)\left[p_U^\tn(u^n)p_{X|UW}^\tn(x^n|u^n,w^n)\right]^2\ind_{\mathcal{A}_{\tau}^{(n)}}(x^n,u^n,w^n)\\
	&=\sum_{u^n,w^n}p_{XUW}^\tn(x^n,u^n,w^n)p_U^\tn(u^n)p_{X|UW}^\tn(x^n|u^n,w^n)\ind_{\mathcal{A}_{\tau}^{(n)}}(x^n,u^n,w^n)\\
	&=\left(p_X^\tn(x^n)\right)^2\sum_{u^n,w^n}p_{UW|X}^\tn(u^n,w^n|x^n)\frac{p_U^\tn(u^n)\cdot p_{X|UW}^\tn(x^n|u^n,w^n)}{p_X^\tn(x^n)}\ind_{\mathcal{A}_{\tau}^{(n)}}(x^n,u^n,w^n)\\
	\label{eq:ind_bound}&\leq \left(p_X^\tn(x^n)\right)^2\cdot 2^{n\tau\beta}\sum_{u^n,w^n}p_{UW|X}^\tn(u^n,w^n|x^n)\left(\frac{p_U^\tn(u^n)\cdot p_{X|UW}^\tn(x^n|u^n,w^n)}{p_X^\tn(x^n)}\right)^{1-\beta}.
	\end{align}
	The equality in~\eqref{eq:additive_var} follows since the encoding is independent for each $u^n$, and  the inequality~\eqref{eq:ind_bound} follows from~\eqref{eq:typ_indbound}. Plugging this bound into~\eqref{eq:g1_bound} yields
	\begin{align}
	\E g_1(A)&\leq \log (e)2^{n\tau\beta}\sum_{x^n,u^n,w^n}p_X^\tn(x^n)p_{UW|X}^\tn(u^n,w^n|x^n)\left(\frac{p_U^\tn(u^n)\cdot p_{X|UW}^\tn(x^n|u^n,w^n)}{p_X^\tn(x^n)}\right)^{1-\beta}\\ 
	&=\log (e)2^{n\tau\beta}\left(\E \left[\left( \frac{p(U)p(X|W,U)}{p(X)}\right)^{1-\beta}\right]\right)^n
	\\&=\log (e)\exp{\left(n \left(\beta\tau +\log \E Z^{1-\beta} \right)\right)}\\
	&=\log (e)\exp{\left(n \left((1-\gamma)\tau +\log \E Z^\gamma \right)\right)} \label{eq:bound_g1}
	\end{align}
	where 
	\begin{align}
	Z \dfn \frac{p(U)p(X|U,W)}{p(X)}. 
	\end{align}
	and for simplicity we substituted $\gamma = 1-\beta$. Thus, if $\log\E{Z^\gamma} < 0$ for some $\gamma\in(0,1)$, then we can set $\tau>0$ small enough such that the bound~\eqref{eq:bound_g1} vanishes as $n\to \infty$. Expanding $\log\E Z^\gamma$ around $\gamma=0$ we have that
	\begin{align} 
	\log\E Z^\gamma &= \gamma\cdot \left. \frac{d\E Z^\gamma}{d\gamma}\right|_{\gamma =0} + O(\gamma^2)\\
	&=\gamma\cdot \E (\log Z) + O(\gamma^2)\\
	&=\gamma \cdot \E \left(\log \left(\frac{p(U)p(X|U,W)}{p(X)}\right)\right) + O(\gamma^2)\\
	&=-\gamma\cdot \left(H(U)-I(X;U,W)\right) + O(\gamma^2), 
	\end{align} 
	which by the assumption in the Lemma is negative for $\gamma>0$ small enough, hence indeed $\E g_1(A)\to 0$ as $n\to\infty$.  
	
	Proceeding to the second and third terms, we have 
	\begin{align} 
	\E g_2(A) &= \E\left[\sum_{x^n}p_2(x^n|A)\log\frac{p_2(x^n|Z)}{p_X^\tn(x^n)}\right]\\
	&\leq n\log\left(1/\min_x p_X(x)\right)\E\left[\sum_{x^n}p_2(x^n|A)\right],
	\end{align}
	and 
	\begin{align}\label{eq:g3_bound}
	\E g_3(A) &\leq h\left(\E \left[\sum_{x^n}p_2(x^n|A)\right]\right),
	\end{align} 	
	where in~\eqref{eq:g3_bound} we have used Jensen's inequality for the binary entropy function. To conclude our proof, it thus suffices to show that 
	\begin{align}\label{eq:suff_g23}
		\E \left[\sum_{x^n}p_2(x^n|A)\right] = o(1/n).
	\end{align} 
	To that end, write:
	\begin{align}
	\E\left[\sum_{x^n}p_2(x^n|A)\right] & = \sum_{x^n,u^n,w^n\not\in \mathcal{A}_{\tau}^{(n)}(x^n,u^n,w^n)}p_{XUW}^\tn(x^n,u^n,w^n) \\
	&=\Pr\left(\prod_{i=1}^n Z_i \geq 2^{n\tau}\right)\\
	&=\Pr\left(\prod_{i=1}^n Z_i^\gamma \geq 2^{n\tau\gamma}\right)\\
	\label{eq:markov_Z}& \leq 2^{-n\tau\gamma}\cdot\E\left(\prod_{i=1}^n Z_i^\gamma \right)\\
	&=  \left(2^{-\tau\gamma}\cdot\E Z^\gamma \right)^n.
	\end{align}	
	where $Z^n\sim P_Z^\tn$, $\gamma$ is an arbitrary positive constant, and Markov's inequality was used in~\eqref{eq:markov_Z}.  Hence, for~\eqref{eq:suff_g23} to hold it suffices to show that $\frac{\log\E Z^\gamma}{\gamma} < 0$ for some $\gamma >0$. Using L'Hospital's rule, this happens for $\gamma>0$ small enough if the derivative of $\log \E Z^\gamma$ is negative at $\gamma=0$, which as we have already seen holds under the conditions in the Lemma. 
	\end{proof}
\bibliography{Distributed_Source_Simulation_With_No_Communication}
\bibliographystyle{IEEEtran}
\end{document}